\documentclass[acmsmall,screen,nonacm,timestamp]{acmart}

\usepackage{etoolbox}\newtoggle{techreport}
\toggletrue{techreport}

\iftoggle{techreport}{\renewcommand\footnotetextcopyrightpermission[1]{}
}{}

\setcopyright{cc}
\setcctype{by}
\acmDOI{10.1145/3763069}
\acmYear{2025}
\acmJournal{PACMPL}
\acmVolume{9}
\acmNumber{OOPSLA2}
\acmArticle{291}
\acmMonth{10}
\received{2025-03-26}
\received[accepted]{2025-08-12}

\usepackage{booktabs}   \usepackage{subcaption} 

\usepackage[mathcal]{eucal}

\usepackage[textsize=footnotesize,colorinlistoftodos,disable]{todonotes}

\newcommand{\note}[1]{\todo[disable,color=green!40]{{#1}}}

\makeatletter
\definecolor{todocitecolor}{rgb}{0,0.8,0.8}

\define@key{todonotes}{cite}[]{\setkeys{todonotes}{disable,prepend,caption={Cite},color=todocitecolor}}\define@key{todonotes}{coq}[]{\setkeys{todonotes}{disable,prepend,caption={Coq},color=yellow}}\define@key{todonotes}{nit}[]{\setkeys{todonotes}{disable,prepend,caption={nit},color=green}}\define@key{todonotes}{review}[]{\setkeys{todonotes}{disable,prepend,caption={Reviews requested},color=cyan}}\define@key{todonotes}{imp}[]{\setkeys{todonotes}{disable,color=red}}

\makeatother

\usepackage{siunitx}
\usepackage{mathtools}
\usepackage{mdframed}
\usepackage{tabularx} \usepackage{ifthen}
\usepackage[english]{babel}
\usepackage[utf8]{inputenc}
\usepackage{changepage}
\usepackage{amsmath}
\usepackage{thmtools}
\usepackage{thm-restate}
\usepackage{stmaryrd}\usepackage{graphicx, xcolor}
\usepackage{hyperref}
\usepackage{cleveref}
\usepackage{listings}
\usepackage{tikz}\usetikzlibrary{arrows,positioning,decorations.pathmorphing,arrows.meta,bending}
\usepackage[inline]{enumitem}
\usepackage{array}
\usepackage{stackengine}
\usepackage{svg}
\usepackage{fontawesome}
\usepackage{adjustbox}

\usepackage{wrapfig} 

\usepackage{xspace}

\graphicspath{{./img/}}

\crefname{lstlisting}{listing}{listings}
\Crefname{lstlisting}{Listing}{Listings}

\AtEndPreamble{\theoremstyle{acmdefinition}

\newtheorem{criterion}[theorem]{Criterion}

\crefname{criterion}{criterion}{criteria}
\Crefname{criterion}{Criterion}{Criteria}
\crefname{property}{property}{properties}
\Crefname{property}{Property}{Properties}
}

\newcommand\scalemath[2]{\scalebox{#1}{\mbox{\ensuremath{\displaystyle #2}}}}

\usepackage{deduction}
\usepackage{dlstalk}

\definecolor{dkbrown}{rgb}{0.4,0.4,0}
\definecolor{dkred}{rgb}{0.6,0,0}
\definecolor{dkgreen}{rgb}{0,0.3,0}
\definecolor{ltgreen}{rgb}{0,0.5,0.1}
\definecolor{dkblue}{rgb}{0,0,0.6}
\definecolor{ltblue}{rgb}{0,0.4,0.4}
\definecolor{dkviolet}{rgb}{0.3,0,0.5}

\lstdefinelanguage{Coq}{texcl=false,
    backgroundcolor={},
morekeywords=[1]{Section, Module, End, Require, Import, Export, Variable,
        Variables, Parameter, Parameters, Axiom, Hypothesis, Hypotheses,
        Notation, Local, Tactic, Reserved, Scope, Open, Close, Bind, Delimit,
        Definition, Let, Ltac, Ltac2, Fixpoint, CoFixpoint, Add, Morphism,
        Relation, Implicit, Arguments, Unset, Contextual, Strict, Prenex,
        Implicits, Inductive, CoInductive, Record, Structure, Canonical,
        Coercion, Context, Class, Global, Instance, Program, Infix, Theorem,
        Lemma, Corollary, Proposition, Fact, Remark, Example, Proof, Goal, Save,
        Qed, Defined, Hint, Resolve, Rewrite, View, Search, Show, Print,
        Printing, All, Eval, Check, Projections, inside, outside, Def},
morekeywords=[2]{forall, exists, exists2, fun, fix, cofix, struct,
        match, with, end, as, in, return, let, if, is, then, else, for, of,
        nosimpl, when},
morekeywords=[3]{Type, Prop, Set, true, false, option, bool, nat, list},
morekeywords=[4]{pose, set, move, case, elim, apply, clear, hnf, intro,
        intros, generalize, rename, pattern, after, destruct, induction, using,
        refine, inversion, injection, inversion\_clear, rewrite, congr, unlock,
        compute, ring, field, fourier, replace, fold, unfold, change,
        cutrewrite, simpl, have, suff, wlog, suffices, without, loss, nat_norm,
        assert, consider, have, cut, trivial, revert, bool_congr, nat_congr,
        symmetry, transitivity, auto, split, left, right, autorewrite},
morekeywords=[5]{by, done, now, exact, reflexivity, tauto, romega, omega,
        assumption, solve, contradiction, discriminate},
morekeywords=[6]{do, last, first, try, idtac, repeat},
morecomment=[s]{(*}{*)},
    morecomment=[s][\ttfamily\color{dkbrown}]{(**}{**)},
    morecomment=[s][\ttfamily\small\color{gray}]{by}{.},
aboveskip=3pt,belowskip=3pt,showstringspaces=false,
morestring=[b]",
tabsize=3,
extendedchars=false,
sensitive=true,
breaklines=false,
basicstyle=\scriptsize,
captionpos=b,
columns=fixed,
    basewidth=3.5pt,
keepspaces=true,
identifierstyle={\ttfamily\color{black}},
keywordstyle=[1]{\ttfamily\color{dkviolet}},
keywordstyle=[2]{\ttfamily\color{dkblue}},
keywordstyle=[3]{\ttfamily\color{ltblue}},
keywordstyle=[4]{\ttfamily\color{dkblue}},
keywordstyle=[5]{\ttfamily\color{dkgreen}\bfseries},
stringstyle=\ttfamily,
commentstyle={\ttfamily\color{ltgreen}},
literate=
    {\\In}{{\(\in\)}}1
    {\\in}{{\(\in\)}}1
    {\\notin}{{\(\not\in\)}}1
    {forall}{{{\color{dkgreen}\(\forall\)}}}1
    {exists}{{{\color{dkred}\(\exists\)}}}1
    {fun}{{{\color{dkblue}\(\lambda\)}}}1
    {\ nat}{{\ \(\mathbb{N}\)}}2
    {\ bool}{{\ \(\mathbb{B}\)}}2
    {<-}{{\(\leftarrow\)}}2
    {=>}{{\(\Rightarrow\)}}2
    {==}{{\texttt{==}\;}}2
    {==>}{{\texttt{==>}\;}}3
    {:>}{{\texttt{:>}}}2
    {:=}{{\(\coloneq\)}}2
    {->}{{\(\rightarrow\)}}2
    {<->}{{\(\leftrightarrow\)}}2
    {<==}{{\(\leq\)}}1
    {<>}{{\(\neq\)}}1
    {<=}{{\(\leq\)}}1
    {|-}{{\(\vdash\)}}1
    {;}{{\texttt{;}}}1
    {:}{{\texttt{:}}}1
    {\#}{{\(^\star\)}}1
    {\\o}{{\(\circ\)}}1
{\/\\}{{\(\wedge\)}}1
    {\\\/}{{\(\vee\)}}1
    {False}{{\(\bot\)}}1
    {True}{{\(\top\)}}1
    {++}{{\texttt{++}}}2
    {[}{{\texttt{[}}}1
    {]}{{\texttt{]}}}1
{)=>}{{\()\!\!\!\Rightarrow\)}}2
    {=(}{{\(=\!\!\!(\)}}2
    {]=>}{{\(]\!\!\!\Rightarrow\)}}2
    {=[}{{\(=\!\!\![\)}}2
    {>->}{{\texttt{>->}\;}}3
{\@\@}{{\(@\)}}1
{n0}{{\texttt{n}\(_0\)}}2
    {n1}{{\texttt{n}\(_1\)}}2
    {n2}{{\texttt{n}\(_2\)}}2
    {N0}{{\texttt{N}\(_0\)}}2
    {N1}{{\texttt{N}\(_1\)}}2
    {N2}{{\texttt{N}\(_2\)}}2
    {i0}{{\texttt{i}\(_0\)}}2
    {i1}{{\texttt{i}\(_1\)}}2
    {i2}{{\texttt{i}\(_2\)}}2
    {h0}{{\texttt{h}\(_0\)}}2
    {h1}{{\texttt{h}\(_1\)}}2
    {h2}{{\texttt{h}\(_2\)}}2
{''}{}0
    {`}{\texttt{\textasciigrave}}1
    {\\mapsto}{{\(\mapsto\;\)}}1
    {\\hline}{{\rule{\linewidth}{0.5pt}}}1
}[keywords,comments,strings]

\newcommand{\change}[3]{#3}
\newcommand{\changeOurs}[3]{#3}
\newcommand{\changeNoMargin}[1]{#1}
\newcommand{\changeOursNoMargin}[1]{#1}
\newcommand{\changeCameraReady}[2]{#2}

\makeatletter
\@ifclasswith{acmart}{review}{\makeatletter
  \patchcmd{\@addmarginpar}{\ifodd\c@page}{\ifodd\c@page\@tempcnta\m@ne}{}{}
  \makeatother
  \reversemarginpar
\newcounter{change}
  \setcounter{change}{0}
\renewcommand{\change}[3]{\refstepcounter{change}\todo[disable,prepend,size=\tiny,color=red!30]{\label{#1}\textbf{(\thechange)} #2}{\color{red}#3}}\renewcommand{\changeOurs}[3]{\refstepcounter{change}\todo[disable,prepend,size=\tiny,color=blue!30]{\label{#1}\textbf{(\thechange)} #2}{\color{blue}#3}}\renewcommand{\changeCameraReady}[2]{\todo[prepend,size=\tiny,color=green]{#1}{\color{green}#2}}\renewcommand{\changeNoMargin}[1]{{\color{red}{#1}}}\renewcommand{\changeOursNoMargin}[1]{{\color{blue}{#1}}}}
\makeatother

\crefname{definition}{Def.{}}{Def.{}}
\Crefname{definition}{Def.{}}{Def.{}}
\crefname{section}{\S}{Sections}
\Crefname{section}{\S}{Sections}
\crefname{figure}{Fig.{}}{Fig.{}}
\Crefname{figure}{Fig.{}}{Fig.{}}
\crefname{proposition}{Prop.{}}{Prop.{}}
\Crefname{proposition}{Prop.{}}{Prop.{}}

\begin{document}

\newcommand{\MyTitle}{Correct Black-Box Monitors for Distributed Deadlock Detection}
\iftoggle{techreport}{\title[\MyTitle\xspace(Technical Report)]{\MyTitle{}: Formalisation and Implementation \\(Technical Report)}{}}{\title{\MyTitle{}: Formalisation and Implementation}}

\author{Radosław Jan Rowicki}
\email{rjro@dtu.dk}
\orcid{0009-0003-0758-722X}
\affiliation{\institution{Technical University of Denmark}
  \city{Kongens Lyngby}
  \country{Denmark}
}

\author{Adrian Francalanza}
\affiliation{\institution{University of Malta}
  \city{Msida}
  \country{Malta}}
\email{adrian.francalanza@um.edu.mt}
\orcid{0000-0003-3829-7391}

\author{Alceste Scalas}
\email{alcsc@dtu.dk}
\orcid{0000-0002-1153-6164}
\affiliation{\institution{Technical University of Denmark}
  \city{Kongens Lyngby}
  \country{Denmark}
}

\newcommand{\toolname}{\texttt{DDMon}\xspace}\newcommand{\Toolname}{\texttt{DDMon}\xspace}

\begin{abstract}
  Many software applications rely on concurrent and distributed (micro)services
that interact via message-passing and various forms of remote procedure calls
(RPC). As these systems organically evolve and grow in scale and complexity, the risk of
introducing deadlocks increases and their impact may worsen: even if only a few services deadlock, many
other services may block \changeCameraReady{Clarify}{while} awaiting responses from the deadlocked ones. As a
result, the ``core'' of the deadlock can be obfuscated by its consequences
on the rest of the system, and diagnosing and fixing the problem can be challenging.

In this work we tackle the challenge by proposing \emph{distributed black-box
monitors} that are deployed alongside each service and detect deadlocks by
only observing the incoming and outgoing messages\changeCameraReady{Mention probes}{, and exchanging \emph{probes}
with other monitors.} We present a formal model that captures popular RPC-based application styles
(e.g., \texttt{gen\_server}s in Erlang/OTP), and a distributed
black-box monitoring algorithm that we prove sound and complete (i.e.,
identifies deadlocked services \changeCameraReady{Lang}{with neither} false positives nor false negatives). We
implement our results in a tool called \Toolname for the monitoring of
Erlang/OTP applications, and we evaluate its performance.

This is the first work that formalises, proves the correctness, and implements
distributed black-box monitors for deadlock detection. Our results are mechanised in Coq. \Toolname is the companion artifact of this paper.

\iftoggle{techreport}{\medskip \noindent \emph{\textbf{NOTE:} This technical report is the extended version of a paper with
      the same title accepted at OOPSLA 2025: \url{https://doi.org/10.1145/3763069}}}{}

 \end{abstract}

\ccsdesc[500]{Theory of computation~Logic and verification}
\ccsdesc[500]{Theory of computation~Process calculi}
\ccsdesc[500]{Theory of computation~Distributed algorithms}
\ccsdesc[300]{Theory of computation~Program verification}
\ccsdesc[300]{Software and its engineering~Software verification}
 
\keywords{distributed system, process calculus, deadlock, monitoring, operational correspondence,
  proof}

\maketitle

\section{Introduction}\label{sec:intro}

Modern software applications are often concurrent and distributed in nature, as part of their implementations often includes ensembles of \emph{(micro)services} that perform specific tasks, run concurrently across a network, and exchange information via \emph{message passing}.
A design pattern prevalent to such systems is that of \emph{remote procedure calls} (RPC),
in which a client sends a request to some service (which acts as an RPC server) and
awaits a response.
In order to compute such a response, the RPC server itself
might require further RPC interactions, thus acting as a client to other
services.
Various development frameworks support this style of
software organisation.
For instance, languages running on the Open Telecom Platform (OTP) like Erlang and Elixir support services based on \emph{behaviours} \changeOurs{change:introduction:1}{Define ``behaviour''}{(i.e., structures defining callback functions that a software module must implement)} representing generic servers and state machines~\cite{erlang-book,erlang-doc-gen-server,erlang-doc-gen-statem},
whereas
the Akka/Pekko frameworks support the development of RPC-based services as actors~\cite{PekkoGrpc}.
Some frameworks even assist programmers in exposing selected objects and APIs to web-based RPC via OpenAPI.\footnote{\url{https://www.openapis.org}, \url{https://openapi-generator.tech}}

\changeOurs{change:introduction:2}{Reword}{
(Micro)service-oriented systems are often preferred to
monolithic systems, due to their flexibility, scalability, and modularity~\cite{DBLP:journals/jss/AksakalliCCT21}.
By adhering to agreed-upon RPC APIs, services can encapsulate their internal logic, allowing them to evolve independently. 
This decoupling fosters both modularity in design and greater adaptability.
As a result, RPC-based service architectures often evolve organically, with RPC calls and services being added, modified, or removed as the application grows and requirements change~\cite{DBLP:journals/jss/AssuncaoKMS23,DBLP:journals/ife/NeriSZB20}.
For example, a commonly used computation can be refactored into a dedicated RPC call, making it reusable across services.
Likewise, frequently recurring interaction patterns involving sequences of requests and responses across different services can be abstracted behind a single RPC call --- potentially coordinating further RPC interactions 
in the background.
}

\begin{wrapfigure}{R}{0.21\textwidth}
  \vspace{-2pt}
  \hspace{-10pt}\scalebox{0.9}{\begin{tikzpicture}[->,>=stealth',shorten >=1pt,auto,node distance=1.2cm,
  semithick, state/.style={circle, draw, minimum size=20pt},every
  label/.style={align=left},every node/.style={align=left}]

\node (start) {};
  \node[state] (p1) [right of=start] {$S_1$};
  \node[state] (p3) [above right of=p1] {$S_3$};
  \node[state] (p2) [below right of=p3] {$S_2$};

\path (start) edge[above] node {1} (p1)
  (p1) edge[below] node {2} (p2)
  (p2) edge[above right] node {3} (p3)
  (p3) edge[above left] node {4} (p1);
\end{tikzpicture}
} \hspace{-5pt}
  \vspace{-12pt}
\end{wrapfigure}
\paragraph{\changeCameraReady{Restored paragraph title}{Deadlocks in RPC-Based Systems.}}RPC-based systems are prone to \emph{deadlocks}.
Concretely, a deadlock may occur when a service \(\Ser_1\) needs an externally-computed result to serve a request: e.g., \(\Ser_1\) may send an RPC request to some service \(\Ser_2\), which in turn sends further RPC requests to other services, inducing a chain of dependencies where some service \(\Ser_n\) sends an RPC request back to \(\Ser_1\) itself, which is still waiting for \(\Ser_2\)'s response.
Thus, the services \(\Ser_1,\ldots,\Ser_n\) end up blocked in a dependency
cycle: see the depiction above, where the arrows represent requests,
and the numbers represent the order in which they are performed.
These deadlocks are more easily introduced when services are independently maintained and evolve over time.
To make matters worse, these deadlocks can be very hard to identify, \changeOurs{change:introduction:3}{Reword}{understand},
and fix.
For instance, the user of an RPC system may only see a subset of the services involved; high amounts of
traffic could make the system temporarily unresponsive and some requests may be aborted after a timeout, making it easy to mistake a real deadlock for transient performance issues.
\changeOurs{change:introduction:4}{More natural phrasing}{Moreover, defects such as deadlocks often arise in an unpredictable manner under specific conditions, and can be difficult to reproduce}
(e.g.~dependent on the contents and/or scheduling of RPC requests; we show one such example in \Cref{fig:deadlock-envelope,fig:deadlock-envelope-code}).
\changeOurs{change:introduction:5}{Style improvement}{Even worse,} other services may become blocked \changeCameraReady{Clarify}{while} awaiting responses from the deadlocked ones, thereby making it even harder to identify the real source of the deadlock.

In principle, the presence of deadlocks in a system could be detected via static
analysis. Several solutions have emerged, e.g., in the realm of process
calculi~\cite{DBLP:journals/mscs/CristescuH16,DBLP:journals/jlap/DardhaP22,DBLP:journals/toplas/Kobayashi98,10.1145/2603088.2603116,DBLP:conf/concur/Kobayashi06,DBLP:journals/iandc/0001L17,NgCCStaticDeadlock2016}.
These solutions are however limited.
Theoretically,  the problem of deadlock-freedom is
generally undecidable, and any sound decidable static analysis must have false
positives, thereby rejecting systems that are deadlock-free.
Practically, static analysis methods require access to the entire source code for inspection, which is occasionally not possible (e.g.~closed-source components).
Furthermore, many distributed systems use a variety of technologies \changeOurs{change:introduction:6}{Clarify}{and programming languages}, some of which are hard to
statically analyse accurately (e.g.~Erlang or Elixir, mentioned above).

\paragraph{The Challenge: Deadlock Detection in RPC Systems via Distributed Black-Box Outline Monitors.}
Runtime monitoring for deadlock detection is an alternative to static analysis that avoids the aforementioned limitations.
In this work, we focus on the problem of designing
and developing monitoring processes that, when deployed \changeOurs{change:introduction:7}{Style improvement}{in RPC-based distributed systems}, can accurately identify
deadlocks at runtime. More precisely, we have
the following requirements for our monitors:

\begin{enumerate}
\item\label{item:req:distrib} they must be deployed in a \emph{distributed}
  fashion, \changeOurs{change:introduction:8}{Ditto}{alongside} each service, without introducing any centralised components
  (which 
\changeOurs{change:introduction:9}{Ditto}{can become bottlenecks}
  in large-scale systems);
\item \label{item:req:black-box-outline} they must be \emph{black-box} and
  \emph{outline}, i.e., detect deadlocks by just observing the incoming and
  outgoing messages of the service they monitor, without relying on
  implementation details of the service internals, or invasive
  modifications (\emph{inlining}) in the service executable;\note{Should we say that our tool is ``morally'' black-box?}
\item\label{item:req:transparent} they must be \emph{transparent}, i.e., the
  presence or absence of monitors must not incorrectly alter the execution of the services;
\item\label{item:req:precise} they must be \emph{precise} in their verdicts,
  i.e., report a deadlock \emph{if and only if} a deadlock occurs
  in the system being monitored. Hence, the monitors must have 
  \changeOurs{change:introduction:10}{Clarify}{no false
  positives (i.e., report non-existent deadlocks) \changeCameraReady{Lang (nor->)}{or} false negatives (i.e., miss the detection of actual deadlocks).}
\end{enumerate}

Requirements~(\ref{item:req:transparent}) and~(\ref{item:req:precise}) are standard
notions of \emph{monitoring correctness} in the field of runtime
monitoring~\cite{Fran-21}.
Several algorithms have been presented in the literature for distributed
deadlock detection, starting in the early 1980s with the seminal work
of~\cite{ChandyMisraHaas83,MitchellMerrit84}
\changeOurs{change:introduction:11}{Style}{and occasionally} include proofs
of correctness reminiscent of our requirement~(\ref{item:req:precise}).
Although these algorithms satisfy requirement~(\ref{item:req:distrib}) (distribution),
their deadlock detection logic is presented as an \emph{inlined} part of a
potentially-deadlocking process, instead of being deployed as a separate process
that monitors the incoming/outgoing messages of an already-existing process.
Consequently, \changeOurs{change:introduction:12}{Style}{any formalisation of such algorithms (when available) cannot
satisfy requirement~(\ref{item:req:black-box-outline}) (black-box outline
monitors) and does not permit one to state and prove
requirement~(\ref{item:req:transparent}) (transparency).}
Still, their correctness results suggest that such algorithms are a good starting point towards
designing and developing distributed, black-box, outline monitors that
satisfy all our requirements~(\ref{item:req:distrib})--(\ref{item:req:precise}).

\change{change:tradeof}{Mention that our requirements are a design choice}{
Of course, the selection of requirements~(\ref{item:req:distrib}) and (\ref{item:req:black-box-outline}) above is a design choice based on trade-offs, and different approaches to deadlock monitoring lead to different trade-offs. For example, a centralised deadlock supervisor may be feasible for small systems (despite not being scalable) and it might take better
advantage of its immediate view of the global state. Offline deadlock detection (via log analysis) may be simpler to implement, but does not allow for immediate intervention upon deadlock detection. Non-black-box (inlined) monitors may better observe the system behaviour, at the cost of being more invasive and dependent on the service internals. In contrast to that, our requirements target large systems which may contain closed-source services which can only be verified at runtime and with weak assumptions (i.e., adherence to an RPC protocol).}

\paragraph{Contributions and Outline of the Paper.}
This work contributes to the field of runtime monitoring in distributed deadlock detection as follows:

\begin{enumerate}
\item\label{item:contrib:general-correctness} In \Cref{sec:correctness} we
  establish general criteria of correctness for deadlock detection monitors in
  distributed systems, based on standard notions of transparency, soundness, and
  completeness in runtime verification literature.
\item\label{item:contrib:srpc-net} \changeOurs{change:contrib:ref1}{Add reference to relevant section}{In \Cref{sec:formal}}
  we provide a general formal model of networks
  of RPC services, focusing on the class of services (that we dub \emph{SRPC},
  for \emph{Single-threaded RPC}) which may result in deadlocks.
\item\label{item:contrib:transparent-instr} \changeOurs{change:contrib:ref2}{Ditto}{In \Cref{sec:monitoring}}
  we formalise a method for the
  \emph{distributed, black-box, outline instrumentation} of such networks with a
  generic class of \emph{proxy monitors}, and prove that this general
  instrumentation guarantees transparency by construction.
\item\label{item:contrib:monitor-algo} \changeOurs{change:contrib:ref3}{Ditto}{In \Cref{sec:algorithm}}
  we formalise a concrete monitoring
  algorithm (inspired by~\cite{ChandyMisraHaas83,MitchellMerrit84}) usable in the
  generic proxy monitors of contribution~(\ref{item:contrib:transparent-instr})
  above, and \changeOurs{change:contrib:ref4}{Ditto}{in \Cref{sec:proof}} we prove that this algorithm detects deadlocks in a sound and
  complete way (i.e., without false positives \changeCameraReady{Lang (nor->)}{and} false negatives).
Despite the initial inspiration, the proofs of correctness of our algorithm
  are quite challenging and radically different from~\cite{ChandyMisraHaas83,MitchellMerrit84}, since we need to formalise and prove many technical results (e.g.~about
  the \emph{instrumentation} of a monitor-free network) and invariants that were
  absent or under-specified in previous work: \change{change:contrib:cex}{Mention subtle counterexample}{in \Cref{sec:related:edge} we discuss a counterexample that arises if such technical results and subtle invariants are overlooked.}
\item\label{item:contrib:impl} \changeOurs{change:contrib:ref5}{Ditto}{In \Cref{sec:tool}}
  we present \Toolname, \changeCameraReady{Moved}{the companion artifact of this work:} a prototype tool that
  applies our distributed black-box monitor formalisation and algorithm to
  systems based on Erlang/OTP and its widely-used
  \texttt{gen\_server} behaviour. We evaluate the performance of \Toolname with different detection strategies.
\end{enumerate}
We discuss related work in \Cref{sec:related} and conclude in \Cref{sec:conclusion}.

To the best of our knowledge, this is the first work that formalises, proves the
correctness, and implements distributed black-box outline monitors for deadlock
detection. Our formalisation and theoretical results are mechanised in
\change{change:coq-version}{Specify Coq version}{Coq 8.20}; the mechanisation is available in~\cite{rowicki_2025_16909482}, and the
mechanised results are marked throughout this paper with the symbol \coqed. \iftoggle{techreport}{}{Additional materials are available in a technical report~\cite{TECHREPORT}.}

 \section{Correctness Criteria for Deadlock Detection Monitors in Distributed Systems}
\label{sec:correctness}

In this section we establish \change{change:correctness:define-title}{Clarify ``correct monitoring'' in the title}{the meaning of ``correct monitoring'' in the title of this paper: we define}
general criteria for assessing the correctness of deadlock detection
monitors in distributed systems, based on widely-accepted standards for runtime
verification~\cite{DBLP:series/lncs/BartocciFFR18}. Some results akin to the \emph{detection soundness}
and \emph{completeness} \changeCameraReady{Connect to content}{defined below} can be found in earlier work on distributed deadlock detection
algorithms~\cite{ChandyMisraHaas83,MitchellMerrit84} but they are not presented in \changeOurs{}{style}{a
monitoring setup}, i.e., as processes that can reach \emph{verdicts} when \emph{instrumented} on
a system, and act \emph{transparently}. 
This section systematises these concepts and aims to serve as a
blueprint for research in the field of deadlock monitoring. It also serves as an overview of our
approach and results (presented in the next sections).

\paragraph{Networks and Deadlocks.}In distributed systems, deadlocks span over multiple agents communicating across a network.
Therefore, we need a formal model of a network \(\Net\) equipped with semantics describing
executions as a series of transitions of the form \(\Net_0 \trans{\NAct} \Net_1\); depending on the
formalisation style, the transition \(\trans{\NAct}\) can represent reduction semantics or LTS
semantics~\cite{RMKeller1976}. We write \(\Net_0 \trans{\ppath} \Net_1\) for a sequence of transitions \(\ppath\),
which we call a \emph{path}.
We also need to define when \(\Net\) \emph{has a deadlock}, i.e., a set of agents that are stuck waiting
for \changeOurs{}{Ditto}{one another's} inputs: the specific formalisation depends on the specifics of the agents.
\emph{(We formalise our networks, agents (called \emph{services}), and their deadlocks in \Cref{sec:formal}.)}

\paragraph{Monitors and Instrumentation.}Monitors are deployed by \emph{instrumenting} them over a network. Therefore, we need some form of
\emph{instrumentation function \(\instr\)} that turns an unmonitored network \(\Net\) into a
monitored network \(\MNet\). Dually, the \emph{deinstrumentation} \(\deinstr(\MNet)\) ``strips off'' the monitors of \(\MNet\) and
returns an unmonitored network.
Instrumented networks have transitions \(\MNet_0 \trans{\MAct} \MNet_1\), extended over paths as
\(\MNet_0 \trans{\mpath} \MNet_1\); notice that \(\MAct\) and \(\mpath\) may differ from the
unmonitored system's \(\Act\) and \(\ppath\) (e.g., \(\MAct\) and \(\mpath\) may include additional
information or messages).
\change{change:instr-mid-state}{Mention that instrumentation can be added at any state of the network}{
In principle, an instrumentation function \(\instr\) can be applied to any state of a network --- but in practice, the deployment of monitors is modelled by instrumenting a network in its initial state.}
The details of instrumentation determine how, and to what extent, the monitors can observe the
network behaviour.
\emph{(We formalise our instrumentation of \emph{black-box outline proxy monitors} in
  \Cref{sec:monitoring}.)}

\paragraph{Transparency.}A highly desirable property of monitor instrumentation is \emph{transparency}
(Criterion~\ref{crit:transparency} below), ensuring that the addition of monitors does not
improperly change the behaviour of the system. Specifically, we formulate transparency as an
\emph{operational correspondence}~\cite{DBLP:journals/iandc/Gorla10} requiring that all reachable
states of a network \(\Net\) are preserved after instrumentation (completeness), and all states reachable by
an instrumented network \changeCameraReady{Fix, was $\MNet$ before}{\(\instr(\Net)\)} are related to a reachable state of \(\Net\) (soundness). This ensures
that instrumentation retains all possible deadlocks of \(\Net\), and does not introduce spurious
ones. \emph{(We prove our transparency results in \Cref{lem:net-transp-compl,lem:net-transp-sound}; recall
  that \(\deinstr(\MNet)\) is \(\MNet\) with monitors ``stripped off.'')}

\begin{criterion}[Monitor Instrumentation Transparency]
  \label{crit:transparency}\label{crit:transparency-completeness}\label{crit:transparency-soundness}A monitor instrumentation \(\instr\) is \emph{transparent} for a network
  \(\Net_0\) \changeCameraReady{More precise}{if and only if} the following clauses hold: \begin{description}
  \item[Completeness:]For any \(\Net_1\) and path \(\ppath\),
    \(\Net_0 \trans{\ppath} \Net_1\) implies \(\exists \mpath,\MNet_1\) such that \(\instr(\Net_0) \trans{\mpath} \MNet_1\) and \(\deinstr(\MNet_1) = \Net_1\).
  \item[Soundness:]For any monitored network \(\MNet_1\) and monitored path \(\mpath_0\), \(\instr(\Net_0) \trans{\mpath_0} \MNet_1\) implies that \(\exists \MNet_2, \Net_2, \mpath_1, \ppath\) such that \(\MNet_1 \trans{\mpath_1} \MNet_2\) and \(\deinstr(\MNet_2) = \Net_2\) and
    \(\Net_0 \trans{\ppath} \Net_2\).
  \end{description}
\end{criterion}

\change{change:transp-sim}{Describe how transparency is related to similarity}{
Intuitively, transparency is achieved if a network \(\Net\) and its instrumented version
\(\instr(\Net)\) are mutually weakly similar; however, the operational
correspondence above allows for more flexibility in how the states of \(\Net\) and \(\instr(\Net)\) are related
(as operational correspondences are typically weaker than simulation relations~\cite{DBLP:journals/corr/PetersG15}).}
Notice that the
``soundness'' clause of \Cref{crit:transparency} allows the instrumented system to reach an
intermediate configuration \(\MNet_1\) which, after some further steps \(\mpath_1\), aligns with an
unmonitored network state \(\Net_2\). Depending on the instrumentation details, transparency may be
stronger (e.g., with \(\MNet_1 = \MNet_2\) and \(\mpath_1\) always empty).

\paragraph{Preciseness.}\Cref{crit:detection-preciseness} below establishes that deadlock detection monitors must
be \emph{precise} in their verdicts, i.e., they must be able to report all possible deadlocks of the
original unmonitored network (\emph{completeness}), and all the deadlocks they report must be real
(\emph{soundness}). \emph{(We prove the preciseness of our approach in \Cref{thm:deadlock-detection-preciseness}; recall
  that \(\deinstr(\MNet)\) is \(\MNet\) with monitors ``stripped off.'')}

\begin{criterion}[Preciseness of Deadlock Detection]
  \label{crit:detection-preciseness}\label{crit:detection-completeness}\label{crit:detection-soundness}An instrumentation \(\instr\) achieves \emph{precise deadlock detection} for a network \(\Net_0\) \changeCameraReady{More precise}{if and only if},
  for any monitored network \(\MNet_1\) and monitored path \(\mpath\) such that \(\instr(\Net_0) \trans{\mpath} \MNet_1\), the following clauses hold:
  \begin{description}
  \item[Completeness:]If there is a deadlock in \(\deinstr(\MNet_1)\), then there are \(\MNet_2, \mpath_1\) such that
    \(\MNet_1 \trans{\mpath_1} \MNet_2\) and \(\MNet_2\) reports a deadlock.\item[Soundness:]If \(\MNet_1\) reports a deadlock, then there is a deadlock in \(\deinstr(\MNet_1)\).
  \end{description}
\end{criterion}

Note that \Cref{crit:detection-preciseness} does not require deadlocks to be reported immediately:
the ``completeness'' clause allows the monitored system to take some additional steps \(\mpath_1\)
after a deadlock occurs; similarly, the ``soundness'' clause allows a deadlock to happen in a state
that precedes \(\MNet_1\) (where the deadlock is reported). Also note that \Cref{crit:detection-preciseness} is only effective in combination with
\Cref{crit:transparency} (transparency): e.g., a non-transparent instrumentation \(\instr(\Net_0)\)
that blocks all transitions of \(\Net_0\) would trivially satisfy
\Cref{crit:detection-preciseness} for any not-yet-deadlocked \(\Net_0\).
Finally, note that the ``completeness'' of \Cref{crit:detection-preciseness} can be 
strengthened by showing that, under suitable \emph{fairness}
assumptions~\cite{Glabbeek_2019}, all deadlocks \emph{will} be eventually
reported. \emph{(We prove this stronger result in \Cref{thm:detect-completeness-upgrade}.)}

 \section{A Formal Model of Networks of RPC-Based Services}\label{sec:formal}

In this section, we formalise networks of \emph{services} which interact
in a style that we dub \emph{Single-threaded Remote Procedure Call (SRPC)},
\changeOurs{}{clarity}{where requests are processed one at a time}. 
SRPC captures common patterns found
in OTP-based systems programmed in Erlang or Elixir, via the widely-used
\texttt{gen\_server} and \texttt{gen\_statem} \emph{behaviours}~\cite{erlang-book}. Similar patterns can be found in Actor-based systems based e.g.~on
Akka/Pekko~\cite{PekkoGrpc}.
We formalise SRPC services in \Cref{sec:formal:serv-syntax-semantics}, we
compose them into networks in \Cref{sec:formal:networks}, and formalise
deadlocks in \Cref{sec:formal:deadlocks}. In \iftoggle{techreport}{\Cref{app:replicated}}{\cite[Appendix B]{TECHREPORT}} we show how
our SRPC-based model can encode patterns like service replication.

\subsection{Services: Syntax and Semantics}
\label{sec:formal:serv-syntax-semantics}

\Cref{def:syntax-service} below introduces the syntax of services.

\begin{definition}[Services]
  \label{def:syntax-service}\label{fig:syntax-service}

  A \emph{service} \(\Ser\) is defined according to the following grammar.

  \noindent \begin{minipage}{0.45\linewidth}
    \centering
    {\small\begin{bnf}
  \bnfProd{Communication tag}{\Chan}\bnfDef \qtag \bnfSep \rtag \bnfSep \ctag \\
  \bnfProd{Name}{\Name}\bnfDef \texttt{n0} \bnfSep \texttt{n1} \bnfSep \ldots \\
  \bnfProd{Client}{\NameE}\bnfDef \Name \bnfSep \nameext \\
\end{bnf}
 }
  \end{minipage}
  \hfill \begin{minipage}{0.5\linewidth}
    {\small\begin{bnf}
  \bnfProd{Process}{\Proc}\bnfDef \ldots
  \\
  \bnfProd{Queue}{\Que}\bnfDef \quenil \bnfSep \queel{\Name}{\Chan}\bnfSep \queapp{\Que}{\Que}\\
  \bnfProd{Service}{\Ser}\bnfDef \pqued{\IQue}{\Proc}{\OQue}\end{bnf}
 }
  \end{minipage}
  \vspace{-2mm}\end{definition}

By \Cref{def:syntax-service}, a service is modelled as a triplet \(\pqued{\IQue}{\Proc}{\OQue}\)
which includes an \emph{input queue} \(\IQue\) \changeOurs{}{style}{storing incoming messages},
an \emph{output queue} \(\OQue\) containing outgoing messages,
and a \emph{process} \(\Proc\) describing the program logic. 
The input/output queues model asynchronous communication with non-blocking send
operations. 
They may contain messages of the form \(\queel{\Name}{\Chan}\), where \(\Name\) is the \emph{name} of the service that is either \changeOurs{}{typo}{the} sender or
recipient (depending on whether the message is in the input or output queue,
respectively), and \(\Chan\) \changeOurs{}{typo}{is} either \(\qtag\) if the message is a
\emph{query}, \(\rtag\) if the message is a \emph{response}, or \(\ctag\) if it is a \emph{cast}.
\change{change:concat}{Clarify the ++ operator for queues}{The symbol \(\queappsym\) denotes standard queue concatenation.}

Our formalisation is streamlined to focus our technical development:
\begin{itemize}
\item Since our results do not depend on the concrete syntax of \(\Proc\) (because
  the process internals are not observable by our black-box monitors
  introduced in \Cref{sec:monitoring}), in \Cref{def:syntax-service} above we
  leave the syntax of the service process \(\Proc\) unspecified, and we identify
  the category of \emph{SRPC services} by only requiring that \(\Proc\)'s
  behaviour is compatible with an \emph{abstract SRPC process} (\Cref{def:srpc-service}
  below). We will
  sometimes use pseudo-code to represent service processes
  (see~\Cref{ex:srpc-process}).
\item In \Cref{def:syntax-service} the data carried by messages is omitted, and
  we only show the sender/recipient name and communication tag (which are needed
  for message dispatching and for services to distinguish between incoming
  queries and responses).
\item \emph{Casts} (marked with the tag \(\ctag\)) are queries that cannot receive a
  response. They allow for ``fire-and-forget'' RPC calls, and they are part of
  the Erlang/OTP \texttt{gen\_server} specification. As shown later (\Cref{fig:semantics-net}), the cast tag \(\ctag\) is used
  only for sending, while delivered casts are turned into queries \(\qtag\) from
  an unspecified sender $\nameext$.
\end{itemize}

\begin{definition}[Abstract and concrete SRPC processes, and SRPC services]
  \label{def:abstract-srpc-proc}\label{def:concrete-srpc-proc}\label{fig:syntax-gen-srpc-state}\label{def:srpc-service}\label{def:simulation}The states of an \defn{abstract SRPC process $\GenProc$} are the following,
  where \(\NameE\) is either a client name \(\Name\) or \(\nameext\)
  (meaning ``no name''):
{\small\begin{bnf}
  \bnfProd{}{\GenProc} \bnfDef \srpcf \hfill\quad\text{\footnotesize(awaiting a query from any client)}\bnfSepN \srpcw{\NameE_c}\hfill\quad\text{\footnotesize(actively serving client \(\NameE_c\))}
  \bnfSepN \srpcl{\NameE_c}{\Name_s}\hfill\quad\text{\footnotesize(awaiting a response from \(\Name_s\) while serving client \(\NameE_c\))}\end{bnf}
 }

  The LTS semantics of \(\GenProc\) is defined by the rules in \Cref{fig:semantics-abstract-proc},
  using the following transition labels:
  {\small\begin{bnf}
  \bnfProd{Process action}{\ProcAct}\bnfDef \recv{\Name}{\Chan}\hfill\quad\text{\footnotesize(receive a message with tag \(\Chan\) from a service named \(\Name\))}\bnfSepN \send{\Name}{\Chan}\hfill\quad\text{\footnotesize(send a message with tag \(\Chan\) to a service named \(\Name\))}\bnfSepN
  \changeOursNoMargin{\tau}\hfill\quad\text{\changeOursNoMargin{\footnotesize(perform an internal computation)}}\end{bnf}
 }

  \changeOurs{change:formal:tau}{Add $\tau$ to $\ProcAct$'s grammar}{}\change{change:formal:simulation}{Define simulation relation}{Let \(\srpcsym\) denote the standard \emph{simulation relation}~\cite{Milner:89:book,Sangiorgi2001TheP}, i.e., the
    union of all relations \(\mathcal{R}\) such that, whenever \((P,Q) \in \mathcal{R}\)
    and $P \trans{\ProcAct} P'$, then $\exists Q'$ such that $Q \trans{\ProcAct} Q'$ and \((P',Q') \in \mathcal{R}\).}

  We say that \(\Proc\) is a \defn{concrete SRPC process} when all the following conditions hold:
  \begin{enumerate}
  \item\label{it:srpc:sim} \(\issrpc{\Proc}{\GenProc}\) --- i.e., \(\Proc\) is simulated by some state \(\GenProc\) of the
      generic SRPC process.\item\label{it:srpc:loop} \(\Proc\) does not terminate --- \change{change:terminate}{Explain termination}{i.e., every $\Proc'$ reachable from $\Proc$ always exposes further transitions.}
  \item\label{it:srpc:ind} \(\Proc\) does not diverge --- \change{change:diverge}{Clarify ``divergence''}{i.e., if \(\issrpcw{\Proc}{\NameE_c}\), then \(\Proc\) must transition to some
      \(\issrpcf{\Proc'}\) after a finite number of transitions.}
\end{enumerate}

  Given a service \(\Ser = \pqued{\IQue}{\Proc}{\OQue}\) (with the syntax from
  \Cref{fig:syntax-service}), we say that \(\Ser\) is an \defn{\emph{SRPC service}}
  if \(\Proc\) is a concrete SRPC process.\end{definition}

By \Cref{def:abstract-srpc-proc}, the abstract SRPC process \(\GenProc\) has the 3 main
computational states depicted in \Cref{fig-srpc-diag}:

\begin{itemize}
\item In state \(\srpcf\), the process awaits a query (i.e., a request) or a
  cast (i.e., a request with sender \(\NameE_c\) set to \(\bot\)) from the input queue; \item In state \(\srpcw{\NameE_c}\), the process is preparing the response
  for its client \(\NameE_c\) (when \(\NameE_c \neq \bot\)), or handling a cast (when \(\NameE_c =
  \bot\)). Only in this state the process is allowed to send messages and perform internal actions
  (\(\tau\));
\item In state \(\srpcl{\NameE_c}{\Name_s}\), the process is awaiting a response
  from \(\Name_s\).
\end{itemize}

We assign an abstract SRPC state \(\GenProc\) to a concrete SRPC process
\(\Proc\) when \(\Proc\) is simulated by \(\GenProc\). Therefore, the syntax of \(\Proc\) is irrelevant, as long as its
queries and responses (with a suitable LTS semantics) are emitted as a refinement of \(\GenProc\).
\Cref{it:srpc:ind} requires that every SRPC process eventually becomes \(\srpcf\) --- otherwise,
a looping process that emits e.g.~an infinite stream of \(\tau\)-actions could
be assigned the state \(\srpcw{\Name_c}\) for \emph{every} \(\Name_c\).\footnote{This requirement does not prevent SRPC processes from being persistently active: a
  process in the \(\srpcw{\nameext}\) state may e.g.~send a cast to itself and
  immediately receive it after becoming \(\srpcf\), thus continuing its operation.}

\begin{figure}
  \providecommand{\rulescale}{0.82}
\begin{deduction}
  \RuleFit[\rulescale]{\RuleSrpcReadyQuery}{}{\srpcf \trans{\recvq{\NameE_c}}\srpcw{\NameE_c}}&&
  \RuleFit[\rulescale]{\RuleSrpcWorkTau}{}{\srpcw{\NameE_c}\trans{\tau}\srpcw{\NameE_c}}&&
  \RuleFit[\rulescale]{\RuleSrpcWorkReply}{\Name_c \neq \bot
  }{\srpcw{\Name_c}\trans{\sendr{\Name_c}}\srpcf }&&
  \RuleFit[\rulescale]{\RuleSrpcWorkReturn}{}{\srpcw{\bot}\trans{\tau}\srpcf }\end{deduction}
\begin{deduction}
  \RuleFit[\rulescale]{\RuleSrpcWorkCast}{}{\srpcw{\NameE_c}\trans{\sendc{\Name}}\srpcw{\NameE_c}}&&
  \RuleFit[\rulescale]{\RuleSrpcWorkQuery}{}{\srpcw{\NameE_c}\trans{\sendq{\Name_s}}\srpcl{\NameE_c}{\Name_s}}&&
  \RuleFit[\rulescale]{\RuleSrpcLock}{}{\srpcl{\NameE_c}{\Name_s}\trans{\recvr{\Name_s}}\srpcw{\NameE_c}}\end{deduction}
 \vspace{-1em}
  \caption{LTS semantics of the abstract SRPC process \(\GenProc\) (\Cref{def:abstract-srpc-proc}).}\label{fig:semantics-abstract-proc}\end{figure}

\begin{figure}[t]
  \begin{minipage}[c]{0.31\textwidth}
    \vspace{-1mm}\scalebox{0.9}{\begin{tikzpicture}[->,>=stealth', shorten >=2pt, auto, node distance=2cm, semithick,srpc/.style={minimum size=2pt, node distance=1.5cm},every label/.style={align=right},every node/.style={align=center}]

  \node[srpc] (f) {\(\srpcf\)};
  \node[srpc] [below of=f](w) {\(\srpcw{\NameE_c}\)};
  \node[srpc] [below of=w](l) {\(\srpcl{\NameE_c}{\Name_s}\)};
\path (f) edge[bend right] node[left ] {\scriptsize\(\recvq{\NameE_c}\)} (w)(w) edge[bend right] node[right] {\scriptsize \(\begin{cases}
      \sendr{\Name_c} & (\text{if}~~\NameE_c \neq \bot)\\
      \tau & (\text{if}~~\NameE_c = \bot)\\
    \end{cases}\)} (f)(w) edge[loop left,  distance=0.7cm] node[right] {\scriptsize\(\tau\)} (w)(w) edge[loop right, distance=1.3cm] node[left ] {\scriptsize\(\sendc{\Name}\)} (w)(w) edge[bend right] node[left ] {\scriptsize\(\sendq{\Name_s}\)} (l)(l) edge[bend right] node[right] {\scriptsize\(\recvr{\Name_s}\)} (w);
\end{tikzpicture}
 \hspace{-5pt}}
\vspace{-3mm}\captionof{figure}{LTS of the abstract SRPC process, based on \Cref{fig:semantics-abstract-proc}.}\label{fig-srpc-diag}\end{minipage}
  \hfill \begin{minipage}[c]{0.66\textwidth}
    \begin{pseudocode}[|P|l|]{Pseudocode of a concrete SRPC process & Abstract SRPC state}
      \psdef{proxy}{server}		 	& \\
      ~~\psrecvqs{client}{x}	& \(\srpcf\) \\
      ~~\pssendq{server}{x}		& \(\srpcw{\texttt{client}}\) \\
      ~~\psrecvrs{server}{y}	& \(\srpcl{\texttt{client}}{\texttt{server}}\) \\
      ~~\psif{client} != \psnil		& \(\srpcw{\texttt{client}}\) \\
      ~~~~\pssendr{client}{y}		& \(\srpcw{\texttt{client}}\) \\
      ~~proxy(server)					& \(\srpcf\) \\
    \end{pseudocode}
    \vspace{-2mm}\captionof{figure}{Example of a concrete SRPC process. The right column shows SRPC state at each line.}
    \label{fig:example-proxy}
  \end{minipage}
\end{figure}

\begin{example}[An SRPC process]
  \label{ex:srpc-process}\Cref{fig:example-proxy} shows the pseudocode of a concrete SRPC process
  usable as control logic for an SRPC service.
The process implements a simple proxy to another service (here
  called \texttt{server}) by forwarding queries and responses from any client. The
  right column describes the abstract SRPC state corresponding to each line
  during the process execution.

  On line 2, the operation \texttt{\psrecvq} blocks until it receives a query from any other service, and
  returns a tuple containing a reference to the client \note{Is ``reference'' clear?}and a payload \texttt{x}. On line 3, the operation \texttt{\psopsend\_\psQ} sends a query to \texttt{server}, with a reference to the
  sender (\texttt{\psself}) and the forwarded payload \texttt{x}. On line 4, the operation \texttt{\psoprecvfrom} blocks until a response is received from
  \texttt{server}. The \texttt{if} statement on line 5 checks whether the client is defined. If that is the case, line 6
  uses \texttt{\psopsend\_\psR} to forward the response \texttt{y} from \texttt{server} to
  \texttt{client}. Finally, line 7 loops.

  Notice that when the client is undefined (\(\bot\)), the \texttt{if} statement on line 5 makes the process skip
  line 6 and immediately loop, thus directly switching into the \(\srpcf\) state. This switch occurs as a
  \(\tau\)-action enabled by \cref{\RuleSrpcWorkReturn}.
\end{example}

\begin{figure}[t]
\providecommand{\rulescale}{0.85}
\begin{deduction}
  \RuleFit[\rulescale]{\RuleQueFind}{}{\queapp{\queel{\Name}{\Chan}}{\Que}\trans{\qspope{\Name}{\Chan}}\Que }&&
  \RuleFit[\rulescale]{\RuleSerIn}{\IQue_1 = \queapp{\queel{\Name}{\Chan}}{\IQue_0}}{\pqued{\IQue_0}{\Proc}{\OQue}\trans{\recv{\Name}{\Chan}}\pqued{\IQue_1}{\Proc}{\OQue}}\hspace{10pt}\RuleFit[\rulescale]{\RuleSerOut}{\OQue_0 = \queapp{\queel{\Name}{\Chan}}{\OQue_1}\quad }{\pqued{\IQue}{\Proc}{\OQue_0}\trans{\send{\Name}{\Chan}}\pqued{\IQue}{\Proc}{\OQue_1}}\hspace{10pt}\RuleFit[\rulescale]{\RuleSerTauTau}{\Proc_0 \trans{\tau} \Proc_1\quad }{\pqued{\IQue}{\Proc_0}{\OQue}\trans{\tautau}\pqued{\IQue}{\Proc_1}{\OQue}}
\end{deduction}
\begin{deduction}
  \RuleFit[\rulescale]{\RuleQueSeek}{\changeNoMargin{\queel{\Name}{\Chan} \!\neq\! \queel{\Name'}{\Chan'}}\;\;\Que \trans{\qspope{\Name'}{\Chan'}} \Que'}{\queapp{\queel{\Name}{\Chan}}{\Que}\trans{\qspope{\Name'}{\Chan'}}
    \queapp{\queel{\Name}{\Chan}}{\Que'}}&&
  \RuleFit[\rulescale]{\RuleSerTauIn}{\IQue_0 \trans{\qspope{\Name}{\Chan}} \IQue_1\quad \Proc_0 \trans{\recv{\Name}{\Chan}} \Proc_1\quad }{\pqued{\IQue_0}{\Proc_0}{\OQue}\trans{\taurecv{\Name}{\Chan}}\pqued{\IQue_1}{\Proc_1}{\OQue}}\hspace{10pt}\RuleFit[\rulescale]{\RuleSerTauOut}{\Proc_0 \trans{\send{\Name}{\Chan}} \Proc_1\quad \OQue_1 = \queapp{\queel{\Name}{\Chan}}{\OQue_0}\quad }{\pqued{\IQue}{\Proc_0}{\OQue_0}\trans{\tausend{\Name}{\Chan}}\pqued{\IQue}{\Proc_1}{\OQue_1}}\end{deduction}
   \vspace{-1em}
  \captionof{figure}{LTS semantics of services (rules \textbf{\texttt{SER-*}}) and selective queues (rules \textbf{\texttt{QUE-*}}).}\label{fig:semantics-service}
\end{figure}

\begin{definition}[Semantics of services]
  \label{def:semantics-service}\label{def:semantics-que}The LTS semantics of a service \(\Ser\) is inductively defined by the rules in
  \Cref{fig:semantics-service}, using the included semantics of message queues, and the following
  labels:
  {\small\begin{bnf}
  \bnfProd{Queue action}{\QueAct}\bnfDef \qspope{\Name}{\Chan}\hfill\quad\text{\footnotesize(dequeue any message with tag \(\Chan\) from a service named \(\Name\))}\\
  \bnfProd{Service action}{\SerAct}\bnfDef \recv{\Name}{\Chan}\hfill\quad\text{\footnotesize(receive a message with tag \(\Chan\) from a service named \(\Name\))}\bnfSepN \send{\Name}{\Chan}\hfill\quad\text{\footnotesize(send a message with tag \(\Chan\) to a service named \(\Name\))}\bnfSepN \taurecv{\Name}{\Chan}\hfill\quad\text{\footnotesize(internally add to the output queue a message with tag \(\Chan\) to a service \(\Name\))}\bnfSepN \tausend{\Name}{\Chan}\hfill\quad\text{\footnotesize(internally pick from the input queue a message with tag \(\Chan\) from a service \(\Name\))}\bnfSepN \tautau \hfill\quad\text{\footnotesize(internal computation)}\end{bnf}
 }
\end{definition}

In \Cref{def:semantics-service}, the semantic rules of message queues (\ref{\RuleQueFind} and
\ref{\RuleQueSeek}) together with the concatenation operator (\(\queappsym\)) accommodate both input
and output queues in the service semantics (\Cref{fig:semantics-service}). Input queues expose a selective FIFO behaviour: the insertion of a message \(\queel{\Name}{\Chan}\)
into queue \(\Que\) always appends the message to the end
(\(\queapp{\Que}{\queel{\Name}{\Chan}}\)), while message reception can be performed either by
selecting the oldest message in the queue (by pattern matching against
\(\queapp{\queel{\Name'}{\Chan'}}{\Que'}\)), or the oldest message with a specific sender and tag
(transition \(\qspope{\Chan}{\Name}\) by \cref{\RuleQueFind,\RuleQueSeek}). This models e.g.~how processes in Erlang and Elixir can selectively receive inputs from their
mailbox.

\subsection{Networks of Services}
\label{sec:formal:networks}

We now formalise a network as a composition of named services.

\begin{definition}[Networks]
  \label{def:network}A \defn{network \(\Net\)} is a function\todo[nit]{Changed from ``mapping'' to make it clear it's
    not partial} from a set of names to services (from \Cref{def:syntax-service}). We write \(\Net(\Name)\) to denote the service named \(\Name\) within \(\Net\).
  We write \(\NetType\) for the set of all networks.
A network \(\Net\) has the LTS semantics in \Cref{fig:semantics-net}, using
  the following labels:
  {\small\begin{bnf}
  \bnfProd{Network action}{\Act}\bnfDef \nettau{\Name}{\ProcAct}\hfill\quad\text{\footnotesize(internal action performed by service \(\Name\))}\bnfSepN \netcomm{\Name_0}{\Name_1}{\Chan}\hfill\quad\text{\footnotesize(communication of a message with tag $\Chan$ from service \(\Name_0\) to \(\Name_1\))}\end{bnf}
 }

  \noindent We write \(\ppath\) to denote a \defn{path}, i.e., a sequence
  of network transition labels \(\Act_1 \pathseq \Act_2 \pathseq \ldots \pathseq
  \Act_n\). We write \(\Net \trans{\ppath} \Net'\) if either \(\ppath\) is empty
  and \(\Net' = \Net\), or \(\ppath = \Act \pathseq \ppath'\) and \(\exists \Net'':
  \Net \trans{\Act} \Net'' \trans{\ppath'} \Net'\).
\end{definition}

By \Cref{def:network} a network can progress by either allowing a service to perform an internal
action, or letting two services exchange a message. Observe that, by \cref{\RuleSerIn} in
\Cref{fig:semantics-service}, every service always has transitions with label
\(\recv{\Name}{\Chan}\) for any \(\Name\) and \(\Chan\); hence, every service can always append
to its input queue any incoming message from any sender. Casts (i.e., messages with the \(\ctag\)
tag) are delivered by \cref{\RuleNetCast} as queries from anonymous sender (\(\recvc\)), \change{change:cast-query}{Clarify the missing ``receive cast'' action}{which by \Cref{def:abstract-srpc-proc} cannot be responded to. A consequence of this design is
  that the action \(\recv{\Name}{\ctag}\) (``receive a cast from $\Name$'') never occurs at runtime,
  as \cref{\RuleNetCast} ``rewrites'' it into \(\recvc\) (``receive a query from anonymous sender $\nameext$''): this lets us treat queries and
  casts uniformly in the \(\srpcwname\) state, as shown in \Cref{ex:srpc-process}.}

\begin{figure}[tbp]
  \begin{minipage}{0.63\linewidth}
    \begin{minipage}{\linewidth}
      \centering
      \providecommand{\rulescale}{0.82}
\begin{deduction}
  \RuleFit[\rulescale]{\RuleNetTau}{\Net(\Name) \trans{\tauof{\ProcAct}} \Ser }{\Net \trans{\nettau{\Name}{\ProcAct}} \Net \subst{\Name}{\Ser}}
  &&
  \RuleTif[\rulescale]{\RuleNetCast}{{\Net (\Name_0) \trans{\sendc{\Name_1}} \Ser_0}\quad {\Net \subst{\Name_0}{\Ser_0} (\Name_1) \trans{\recvc} \Ser_1}}{\Net \trans{\netcomm{\Name_0}{\Name_1}{\ctag}} \Net \subst{\Name_0}{\Ser_0}\subst{\Name_1}{\Ser_1}}
\end{deduction}
\begin{deduction}
  &&
  \RuleTif[\rulescale]{\RuleNetCom}{\Chan \neq \ctag\quad {\Net (\Name_0) \trans{\send{\Name_1}{\Chan}} \Ser_0}\quad {\Net \subst{\Name_0}{\Ser_0} (\Name_1) \trans{\recv{\Name_0}{\Chan}} \Ser_1}}{\Net \trans{\netcomm{\Name_0}{\Name_1}{\Chan}} \Net \subst{\Name_0}{\Ser_0}\subst{\Name_1}{\Ser_1}}
  &&
\end{deduction}
       \vspace{-1em}
      \captionof{figure}{LTS semantics of networks.}\label{fig:semantics-net}
    \end{minipage}\smallskip\\
    \begin{minipage}{\linewidth}
      \centering
      \scalebox{0.75}{\begin{tikzpicture}[->,>=stealth', shorten >=2pt,auto,node distance=2.5cm, semithick,ser/.style={circle, draw, minimum size=10pt, node distance=3cm},every label/.style={align=left},every node/.style={align=center}]

\node (start) {};
  \node[ser] (server) {\(\strut{\texttt{server}}\)};
  \node[ser] (proxy) [right of=server] {\(\strut{\texttt{proxy}}\)};

\path
  (server) edge[bend left] node[above] {1} (proxy)
  (proxy) edge[bend left] node[below] {2} (server)
  ;
\end{tikzpicture}
 }
\vspace{-1em}
      \captionof{figure}{A simple deadlock in which a \texttt{server} indirectly calls itself through a
      \texttt{proxy} (from \Cref{fig:example-proxy}).
      }\label{fig:deadlock-loop}\end{minipage}
  \end{minipage}
  \hfill
  \begin{minipage}{0.35\linewidth}
    \centering
    \scalebox{0.75}{\begin{tikzpicture}[->,>=stealth', shorten >=2pt, auto, node distance=2cm, semithick,ser/.style={circle, draw, minimum size=10pt, node distance=2cm},init/.style={node distance=1.5cm, minimum size=10pt, color=gray},initedge/.style={decorate, decoration={snake, amplitude=0.5mm, post length=3mm}, gray},every label/.style={align=left},every node/.style={align=center}]

  \node[ser] (s2) {\(E_2\)};
  \node[ser] (s1) [left of=s2] {\(E_1\)};
  \node[ser] (s3) [right of=s2]  {\(E_3\)};

  \node[init] (i1) [below of=s1] {\(\nameext\)};
  \node[init] (i2) [below of=s2] {\(\nameext\)};
  \node[init] (i3) [below of=s3] {\(\nameext\)};

  \node[ser] (p2) [above of=s1] {\(P_2\)};
  \node[ser] (p3) [above of=s2] {\(P_3\)};
  \node[ser] (p1) [above of=s3] {\(P_1\)};

\path
  (i1) edge[initedge] node[near start] {1} (s1)
  (i2) edge[initedge] node[near start] {2} (s2)
  (i3) edge[initedge] node[near start] {3} (s3)

  (s1) edge node[near start] {4} (p2)
  (s2) edge node[near start, right] {5} (p3)
  (s3) edge node[near start, right] {6} (p1)

  (p2) edge node[near start, above right] {7} (s2)
  (p3) edge node[near start, above=5pt] {8} (s3)
  (p1) edge node[near start, right=10pt] {9} (s1)
  ;
\end{tikzpicture}
 }
    \captionof{figure}{Non-deterministic deadlock from \Cref{ex:deadlock-envelope},
      caused by the code in \Cref{fig:deadlock-envelope-code} forming a
      non-trivial dependency cycle. Arrows indicate directions in which queries
      (identified and ordered by numbers) are sent.\label{fig:deadlock-envelope}}\end{minipage}
\end{figure}

\subsection{Locks and Deadlocks}
\label{sec:formal:deadlocks}

In message-passing systems, an agent is \changeOurs{}{style}{generally 
considered “locked” when it cannot progress because it is waiting to receive a specific message}.
In this work we consider a service locked specifically when it is awaiting a response from another service,
as defined below.

\begin{definition}[Lock]\label{def:lock}
  Consider an SRPC service \(\Ser = \pqued{\IQue}{\Proc}{\OQue}\). We say that \defn{\(\Ser\) is locked on name \(\Name_s\)} when the following
  clauses hold:
  \begin{enumerate}
  \item\label{item:def:lock:proc-locked} \(\issrpcl{\Proc}{\NameE_c}{\Name_s}\) \quad\footnotesize{(i.e., the process in
    \(\Ser\) is serving client \(\NameE_c\) and awaits a response from \(\Name_s\))}
  \item\label{item:def:lock:no-response} \(\IQue \notrans{\qspopr{\Name_s}}\) \quad\footnotesize{(i.e., there is no response from \(\Name_s\) in \(\Ser\)'s input queue)}
  \item\label{item:def:lock:no-out} \(\OQue \notrans{\Act}\) \footnotesize{\quad(i.e., the output queue of \(\Ser\) is empty)}
  \end{enumerate}
\end{definition}

Clause~\ref{item:def:lock:proc-locked} of \Cref{def:lock} requires the service
process to await a response from the server \(\Name_s\) (as in \Cref{def:abstract-srpc-proc}).
Moreover, by clause~\ref{item:def:lock:no-out}, the output queue of the service must be empty:
hence, the service may \emph{not} be considered locked as soon as its process becomes locked ---
but the service \emph{may} become locked if its output queue becomes empty before the awaited response
arrives on the input queue (by clause~\ref{item:def:lock:no-response}).
Note that \changeCameraReady{Clarify}{an SRPC service $\Ser$ locked on a service named \(\Name_s\)} can still add incoming messages to its input queue (by \cref{\RuleSerIn} in \Cref{fig:semantics-service}),
but cannot produce outgoing messages until it receives a response \changeCameraReady{Clarify}{from \(\Name_s\).}
Using \Cref{def:lock} we can now adapt the definition of deadlock from \cite{ChandyMisraHaas83}.

\begin{definition}[Deadlock]\label{def:deadlock}
  A non-empty set of service names \defn{\(\mathcal{D}\) is a deadlocked set in \(\Net\)} when for every
  \(\Name \in \mathcal{D}\), the service \(\Net(\Name)\) is locked on a member of \(\mathcal{D}\). \note{This definition used to say ``exclusively locked'' but it may be confusing, as services can
    clearly be locked on one name at a time.}
We say that a \defn{network \(\Net\) has a deadlock} if there is a deadlocked
  set in \(\Net\). We say that a \defn{service named \(\Name\) is deadlocked in \(\Net\)}
  when \(\Name\) belongs to a deadlocked set in \(\Net\).
\end{definition}

Note that \Cref{def:deadlock} captures all ``localised deadlocks,''
\emph{i.e.}, it does not require the entire system to be deadlocked. As
one would expect, deadlocks defined this way are persistent, by
\Cref{lem:deadlock-persistent}.

\begin{lemma}[Persistence of deadlocks]
  \label{lem:deadlock-persistent}
  If\, \(\mathcal{D}\) is a deadlocked set in \(\Net\) and \(\Net \trans{\ppath} \Net'\), then
  \(\mathcal{D}\) is also a deadlocked set in \(\Net'\).
\end{lemma}
\begin{proof}\renewcommand{\qedsymbol}{\coqed}
  We first consider the case \(\Net \trans{\NAct} \Net'\) (with a single action)
  and show that if a service \(\Name \in \mathcal{D}\) is locked on \(\Name'\) in
  \(\Net\), then \(\Name\) is also locked on \(\Name'\) in \(\Net'\). This comes from the fact that the
  only way for \(\Name\) to unlock is to receive a response from \(\Name'\) --- but this cannot happen
  because \(\Name'\) is locked in \(\Net\) as well, hence \(\Name'\) cannot send. We then extend this result to a path \(\mpath\) by induction over \(\ppath\),
  and the statement holds by \Cref{def:deadlock}.
\end{proof}

\begin{example}[A simple deadlock]\label{ex:deadlock-proxy}Recall the proxy program in \Cref{fig:example-proxy}. If the proxy receives a
  query from its own \texttt{server} which is an SRPC service, then
  the server will remain locked until the proxy sends back a response
  (line 6).
However, before that happens, the proxy sends a query to the \texttt{server}
  itself (line 3) and becomes locked on it (line 4). At this point, both services expect
  a response from each other, creating the deadlock depicted in
  \Cref{fig:deadlock-loop}.
\end{example}

\begin{figure}[t]
  \newcommand{\exampleEntry}{E}
\newcommand{\examplePing}{P}

\begin{minipage}[t]{0.66\textwidth}
  \begin{pseudocode}[|P|l|]{Endpoint code & SRPC state}
    \psdef{endpoint}{} & \\
    ~~\psrecvqs{client}{msg} & \(\srpcf\) \\
    ~~\psif{msg == ping} & \\
    ~~~~\pssendr{client}{pong} & \(\srpcw{\texttt{client}}\) \\
    ~~\pselse & \\
    ~~~~\pssendq{msg}{ping} & \(\srpcw{\texttt{client}}\) \\
    ~~~~\psrecvrs{msg}{x} & \(\srpcl{\texttt{client}}{\texttt{msg}}\) \\
    ~~~~\psif{client} != \psnil & \(\srpcw{\texttt{client}}\) \\
    ~~~~~~\pssendr{client}{x} & \(\srpcw{\texttt{client}}\) \\
    ~endpoint() & \\ \end{pseudocode}
\end{minipage}\hfill \begin{minipage}[t]{0.33\textwidth}
  \begin{pseudocode}{System init code}
    \pslet{e1}{\psspawn{endpoint()}}\\
    \pslet{e2}{\psspawn{endpoint()}}\\
    \pslet{e3}{\psspawn{endpoint()}}\\
    \pslet{p1}{\psspawn{proxy(e2)}}\\
    \pslet{p2}{\psspawn{proxy(e3)}}\\
    \pslet{p3}{\psspawn{proxy(e1)}}\\
    \\
    \pssendq{e1}{p1}\\
    \pssendq{e2}{p2}\\
    \!\!\!\pssendq{e3}{p3}\\ \end{pseudocode}
\end{minipage}
 \vspace{-2mm}\caption{Code for \Cref{ex:deadlock-envelope}, possibly leading the deadlock
    in \Cref{fig:deadlock-envelope}. (The code of \texttt{proxy} is in
    \Cref{fig:example-proxy}.)\label{fig:deadlock-envelope-code}}\end{figure}

\begin{example}[A more complex, non-deterministic deadlock]
  \label{ex:deadlock-envelope}
  In \Cref{fig:deadlock-envelope-code,fig:deadlock-envelope}, several parallel
  queries may create a deadlock, depending on how they are delivered to the
  target service.
\Cref{fig:deadlock-envelope-code} (left) defines an ``endpoint''\note{Find a more
  descriptive name} process which handles queries by either immediately
  responding to \texttt{ping}s, or pinging another service specified in the query
  payload.
\Cref{fig:deadlock-envelope-code} (right) initialises a system 
  by spawning three ``endpoint'' services (\(E_1,E_2,E_3\))
  and three proxy services (\(P_1,P_2,P_3\), with the process code from
  \Cref{fig:example-proxy}): these proxies forward anything they receive to the
  respective endpoint service. On lines 7--9, a query is sent to each endpoint service, asking it to ping the
  proxy of another endpoint in a cyclic way.
In this system a deadlock can occur depending on scheduling: if queries are
  delivered in the order in \Cref{fig:deadlock-envelope}, the system
  deadlocks --- but if the queries are delivered in the order 1,4,7,2,5,8,3,6,9 and
  responses are sent as soon as possible, then every query receives a
  response, without deadlocks.
\end{example}

 \section{A Model of Generic Distributed Black-Box Outline Monitors}\label{sec:monitoring}

Generally speaking, a \emph{monitor} is a process that observes a system under
scrutiny and potentially reaches a \emph{verdict} reporting correct or
incorrect behaviours.
In this section, we introduce our model of \emph{black-box outline monitors}
that are deployed alongside each service \(\Ser\) in a network \(\Net\), and
observe the incoming and outgoing messages of \(\Ser\) without depending on the
internals of \(\Ser\): this design allows our monitor model to be adaptable to
different application scenarios. In this work, we specify each monitor to act as a \emph{proxy} that intercepts, examines, and
forwards the incoming and outgoing messages of the overseen service. \note{Mention that we discuss alternative designs in \Cref{sec:conclusion}?}Moreover, we allow our monitors to communicate directly with each other: this is
necessary for reaching deadlock verdicts (since a deadlock may involve an
arbitrary number of services across the network) while maintaining the
\emph{distributed} nature of the system (i.e., our monitors do not require any
centralised components).

We formalise monitored services and networks
(\Cref{sec:monitoring:services-nets}) and monitor instrumentation
(\Cref{sec:monitoring:instr}), proving that they guarantee our correctness
\Cref{crit:transparency}, i.e., \emph{transparency}
(\Cref{lem:net-transp-compl,lem:net-transp-sound}). Note that these results are independent from the specific deadlock detection
monitoring algorithm in use; we present and verify our algorithm in
\Cref{sec:algorithm}.

\subsection{Monitored Services and Networks}\label{sec:monitoring:services-nets}

\Cref{def:syntax-mon} models a \emph{monitored service}
by ``wrapping'' a service with a proxy monitor with its own message queue. \Cref{def:syntax-mon} mentions a \emph{monitor state \(\MState\)}, \emph{probe
\(\Probe\)}, and \emph{monitor algorithm function \(\handle\)}: we leave them unspecified for now (we instantiate them for our
deadlock detection algorithm in \Cref{sec:algorithm}).

\begin{definition}[Monitored services]
  \label{def:syntax-mon}\label{def:semantics-mon}A \defn{monitored service \(\MSer\)} is defined by the following grammar.

  \noindent \begin{minipage}{\linewidth}
    \begin{bnf}
  \bnfProd{Monitor queue element\;}{\MQueElem}\bnfDef \recv{\Name}{\Chan} \bnfSep \send{\Name}{\Chan} \quad\hfill{\text{\footnotesize{(message with tag \(\Chan\) from/to the service named \(\Name\))}}}\bnfSepN \mrecv{\Name}{\Chan}{\Probe}
  \bnfSep \msend{\Name}{\Chan}{\Probe} \quad\hfill{\text{\footnotesize{(probe \(\Probe\) from/to the monitor of the service named \(\Name\))}}}\\
  \bnfProd{Monitor queue\;}{\MQue}\bnfDef \quenil \bnfSep \MQueElem \bnfSep \queapp{\MQue}{\MQue}\\
  \bnfProd{Monitored service\;}{\MSer}\bnfDef \mqued{\MQue}{\MState}{\Ser}\end{bnf}
   \end{minipage}

  The LTS semantics of a monitored service \(\MSer\) is inductively defined in \Cref{fig:semantics-mon}
  using the labels:
  {\small\begin{bnf}
  \bnfProd{Monitored service action\;}{\MSerAct}\bnfDef \SerAct \quad\hfill{\text{\footnotesize{(visble or internal service action, from \Cref{def:semantics-service})}}}\bnfSepN \mrecv{\Name}{\Chan}{\Probe}
  \bnfSep \msend{\Name}{\Chan}{\Probe} \quad\hfill{\text{\footnotesize{(probe \(\Probe\) from/to the monitor of the service named \(\Name\))}}}\bnfSepN \mtaurecv{\Name}{\Chan}\bnfSep \mtausend{\Name}{\Chan} \quad\hfill{\text{\footnotesize{(internal dequeuing/enqueuing via the monitor queue)}}}\bnfSepN \mtaumrecv{\Name}{\Chan}{\Probe} \quad\hfill{\text{\footnotesize{(internal dequeuing of an incoming probe \(\Probe\) from \(\Name\))}}}\end{bnf}
 }
\end{definition}

By \Cref{def:syntax-mon}, a monitored service is a triplet
\(\mqued{\MQue}{\MState}{\Ser}\) comprising a service \(\Ser\)
(from \Cref{def:syntax-service}), a monitor queue \(\MQue\), and a monitor state
\(\MState\). The queue \(\MQue\) is a list of elements which can be either incoming or outgoing
messages, from/to either the monitor \(\MState\) or the service \(\Ser\);
specifically, \(\MState\) may send/receive probes \(\Probe\) to/from other
monitors.
\Cref{fig:example-mon-service} contains a depiction of the semantics of a
monitored service, according to the rules in \Cref{fig:semantics-mon}. 
The intuition is that a monitor \(\MState\) reacts to messages in its input queue
\(\MQue\) and processes them in FIFO order, updating its state and possibly
emitting probes, according to a \emph{monitor algorithm function} \(\handle\) \change{change:of-type}{Write ``of type'' instead of \(:\)}{of type} \(\MStateType \times \MQueElemType \to \MStateType \times \MQueType\) --- where \(\MStateType\), \(\MQueElemType\), and \(\MQueType\) are the sets
of all monitor states, monitor queue elements, and monitor queues, respectively. We assume that \(\handle\) is total, i.e., any call \(\handle(\MState, \MQueElem)\) returns some monitor state \(\MState'\) and queue \(\MQue'\);
\changeCameraReady{Restrict added queue to probe sends}{ moreover, the returned queue $\MQue'$ may only contain outgoing probes (\(\msend{\Name}{\Chan}{\Probe}\)).}
\Cref{\RuleMonIn,\RuleMonMonIn} in \Cref{fig:semantics-mon} say that
incoming messages and probes are appended to the monitor queue. 
By \cref{\RuleMonTauIn,\RuleMonTauMonIn,\RuleMonOut}, if the element at the front
of the monitor queue is either an incoming service message
(\(\recv{\Name}{\Chan}\)), or an incoming probe
(\(\mrecv{\Name}{\Chan}{\Probe}\)), or an outgoing service message
(\(\send{\Name}{\Chan}\)), then that \changeOurs{}{clarity}{element is forwarded as needed while also 
being passed to the monitor function \(\handle\) together with the current monitor
state \(\MState_0\).}
\changeOursNoMargin{Subsequently, \(\handle\) returns} the updated monitor state \(\MState_1\) and a
(possibly empty) list of outgoing probes \(\MQue'\), which are added at the front of the monitor
queue (hence, the probes in \(\MQue'\) pre-empt previously-queued elements). \Cref{\RuleMonMonOut} forwards to the network an outgoing probe if it is at the top of the monitor
queue. \Cref{\RuleMonTau} allows the monitored service to perform internal actions, while
\cref{\RuleMonTauOut} moves an output message from the service into the monitor queue.

\begin{figure}[tbp]
  \centering
  \noindent
  \begin{minipage}{\textwidth}
    \providecommand{\rulescale}{0.8}
\begin{deduction}
  &&
  \Rule[\rulescale]{\RuleMonTauIn}{(\MState_1, \MQue') = \handle(\MState_0, \recv{\Name}{\Chan})\quad \Ser_0 \trans{\recv{\Name}{\Chan}} \Ser_1\quad }{\mqued{\queapp{\recv{\Name}{\Chan}}{\MQue_1}}{\MState_0}{\Ser_0}\trans{\mtaurecv{\Name}{\Chan}}\mqued{\queapp{\MQue'}{\MQue_1}}{\MState_1}{\Ser_1}}
  &&
  \Rule[\rulescale]{\RuleMonOut}{\MQue_0 = \queapp{\send{\Name}{\Chan}}{\MQue_1}\quad (\MState_1, \MQue') = \handle(\MState_0, \send{\Name}{\Chan})\quad }{\mqued{\MQue_0}{\MState_0}{\Ser}\trans{\send{\Name}{\Chan}}\mqued{\queapp{\MQue'}{\MQue_1}}{\MState_1}{\Ser}}
  &&
\end{deduction}
 \end{minipage}\smallskip\\\smallskip
  \begin{minipage}{0.66\textwidth}
    \providecommand{\rulescale}{0.8}
\begin{deduction}
  &&
  \Rule[\rulescale]{\RuleMonIn}{\MQue_1 = \queapp{\MQue_0}{\recv{\Name}{\Chan}}}{\mqued{\MQue_0}{\MState}{\Ser}\trans{\recv{\Name}{\Chan}}\mqued{\MQue_1}{\MState}{\Ser}}\hspace{5pt}
  \Rule[\rulescale]{\RuleMonMonOut}{\MQue_0 = \queapp{\msend{\Name}{\Chan}{\Probe}}{\MQue_1}\quad }{\mqued{\MQue_0}{\MState}{\Ser}\trans{\msend{\Name}{\Chan}{\Probe}}\mqued{\MQue_1}{\MState}{\Ser}}&&
\end{deduction}
\begin{deduction}
  \RuleFit[\rulescale]{\RuleMonMonIn}{\MQue_1 = \queapp{\MQue_0}{\mrecv{\Name}{\Chan}{\Probe}}}{\mqued{\MQue_0}{\MState}{\Ser}\trans{\mrecv{\Name}{\Chan}{\Probe}}\mqued{\MQue_1}{\MState}{\Ser}}\hspace{5pt}
  \RuleTif[\rulescale]{\RuleMonTauMonIn}{\MQue_0 = \queapp{\mrecv{\Name}{\Chan}{\Probe}}{\MQue_1}\quad (\MState_1, \MQue') = \handle(\MState_0, \mrecv{\Name}{\Chan}{\Probe})\quad }{\mqued{\MQue_0}{\MState_0}{\Ser}\trans{\mtaumrecv{\Name}{\Chan}{\Probe}}\mqued{\queapp{\MQue'}{\MQue_1}}{\MState_1}{\Ser}\hspace{5pt}}\end{deduction}
\begin{deduction}
  &&
  \RuleFit[\rulescale]{\RuleMonTau}{\Ser_0 \trans{\tauof{\SerAct}} \Ser_1\quad }{\mqued{\MQue}{\MState}{\Ser_0}\trans{\tauof{\SerAct}}\mqued{\MQue}{\MState}{\Ser_1}}\hspace{5pt}\RuleFit[\rulescale]{\RuleMonTauOut}{\MQue_1 = \queapp{\MQue_0}{\send{\Name}{\Chan}}\quad \Ser_0 \trans{\send{\Name}{\Chan}} \Ser_1}{\mqued{\MQue_0}{\MState}{\Ser_0}\trans{\mtausend{\Name}{\Chan}}\mqued{\MQue_1}{\MState}{\Ser_1}}&&\end{deduction}
 \vspace{-1em}\caption{LTS semantics of monitored services. In \cref{\RuleMonTau}, \(\SerAct\)
      is either \(\recv{\Name}{\Chan}\), \(\send{\Name}{\Chan}\), or \(\tau\) (from~\Cref{def:semantics-service}).}\label{fig:semantics-mon}\end{minipage}\hfill
  \begin{minipage}{0.32\textwidth}
    \vspace{5mm}\scalebox{0.8}{\hspace{-5pt}\begin{tikzpicture}[->,>=stealth',auto,node distance=2.5cm, semithick,mon/.style={node distance=3cm},que/.style={node distance=0.75cm},ique/.style={node distance=1cm},proc/.style={node distance=1cm},oque/.style={node distance=1cm},every label/.style={align=left},every node/.style={align=center}]

\node (start) [] {};\node (init) [above of=start, node distance=0.7cm] {};\node (end) [below of=start, node distance=0.7cm] {};\node[mon]  (m) [right of=start, node distance=1.7cm] {\(\strut{\mquedsepl \MState \mquedsepr} \)};\node[que]  (q) [left of=m] {\(\strut{\mquedl \MQue}\)};\node[ique] (i) [right of=m] {\(\strut{\pquedl \IQue}\)};\node[proc] (p) [right of=i] {\(\strut{\pquedsepl \Proc \pquedsepr }\)};\node[oque] (o) [right of=p] {\(\strut{\OQue \pquedr \mquedr}\)};

\path (init) edge[out=0, in=110, above] node {\(\recvq{\Name}\)} (q)(q) edge[out=60, in=110, dotted] node[above] {} (m)(q) edge[out=60, in=120, line width=1.8pt] node[above] {\footnotesize{\(\mtaurecvq{\Name}\)}} (i)(i) edge[out=60, in=110, dashed] node[above] {\footnotesize{\(\taurecvq{\Name}\)}} (p)(p) edge[out=60, in=110, dashed] node[above right] {\footnotesize{\(\tausendr{\Name}\)}} (o)(o) edge[out=230, in=310] node[above] {\footnotesize{\(\mtausendr{\Name}\)}} (q)(q) edge[out=270, in=250, dotted] node[below] {} (m)(q) edge[out=270, in=0, line width=1.8pt] node[below=2pt] {\(\sendr{\Name}\)} (end);
\end{tikzpicture}
 }
\caption{Visualisation of the communications in a monitored service that
      immediately responds to an incoming query.}
    \label{fig:example-mon-service}\end{minipage}
\end{figure}

\begin{example}
  The arrows in \Cref{fig:example-mon-service} outline how message exchanges occur (in clockwise order) as described by
rules~\ref{\RuleMonIn}, \ref{\RuleMonTauIn}, \ref{\RuleMonTau} (twice), \ref{\RuleMonTauOut}, and
  \ref{\RuleMonOut} in \Cref{fig:semantics-mon}. \Cref{\RuleMonTau} is applied twice (dashed arrows) following \cref{\RuleSerTauIn,\RuleSerTauOut} in
  \Cref{def:semantics-service}. The depiction illustrates the black-box design of the monitor: in fact, \(\MState\) can only
  indirectly observe the transitions drawn with solid arrows, because they result in an incoming or
  outgoing message being added to the monitor queue. Then, the monitor forwards those messages toward the overseen service or to the network (thicker
  arrows) while updating its own state (dotted arrows).\end{example}

We now formalise a monitored network as a composition of named monitored services, with syntax and
semantics similar to \Cref{def:network}.

\begin{definition}[Monitored networks]
  \label{def:mon-network}\label{def:mon-network-semantics}A \defn{monitored network} \(\MNet\) is a function from a set of names to monitored services (from \Cref{def:syntax-mon}). We write \(\MNet(\Name)\) to denote the monitored service named \(\Name\) within \(\MNet\). We write \(\MNetType\) for the set of all monitored networks.
A monitored network \(\MNet\) has the LTS semantics in \Cref{fig:semantics-mon}, using
  the following labels:
  {\small\begin{bnf}
  \bnfProd{Monitored network action}{\MAct}\bnfDef \Act \hfill\quad\text{\footnotesize(visible or internal network action, by \Cref{def:network})}\bnfSepN \mnettau{\Name}{\ProcAct}\hfill\quad\text{\footnotesize(internal action performed by monitor of service \(\Name\))}\bnfSepN \mnetcomm{\Name_0}{\Name_1}{\Probe}\hfill\quad\text{\footnotesize(communication of a probe \(\Probe\) from the monitor of \(\Name_0\) to \(\Name_1\))}\end{bnf}
 }

  \noindent We write \(\mpath\) to denote a \defn{monitored path}, i.e., a sequence
  of monitored network actions \(\MAct_1 \cdot \MAct_2 \cdot \ldots \cdot
  \MAct_n\). We write \(\MNet \trans{\mpath} \MNet'\) if either \(\mpath\) is empty
  and \(\MNet' = \MNet\), or \(\mpath = \MAct \cdot \mpath'\) and \(\exists \MNet'':
  \MNet \trans{\MAct} \MNet'' \trans{\mpath'} \MNet'\).
\end{definition}

\begin{figure}
  \providecommand{\rulescale}{0.8}
\begin{deduction}
  \Rule[\rulescale]{\RuleMNetSCom}{{\MNet (\Name_0) \trans{\send{\Name_1}{\Chan}} \MSer_0}\quad {\MNet \subst{\Name_0}{\MSer_0} (\Name_1) \trans{\recv{\Name_0}{\Chan}} \MSer_1}\quad \Chan \neq \ctag }{\Net \trans{\netcomm{\Name_0}{\Name_1}{\Chan}}\Net \subst{\Name_0}{\Ser_0}\subst{\Name_1}{\Ser_1}}&&
  \Rule[\rulescale]{\RuleMNetSCast}{{\MNet (\Name_0) \trans{\sendc{\Name_1}} \MSer_0}\quad {\MNet \subst{\Name_0}{\MSer_0} (\Name_1) \trans{\recvc} \MSer_1}}{\Net \trans{\netcomm{\Name_0}{\Name_1}{\Chan}}\Net \subst{\Name_0}{\Ser_0}\subst{\Name_1}{\Ser_1}\quad }\end{deduction}
\begin{deduction}
  \RuleTif[\rulescale]{\RuleMNetCom}{{\MNet (\Name_0) \trans{\msend{\Name_1}{\Chan}{\Probe}} \MSer_0}\quad {\MNet \subst{\Name_0}{\MSer_0} (\Name_1) \trans{\mrecv{\Name_0}{\Chan}{\Probe}} \MSer_1}}{\MNet \trans{\mnetcomm{\Name_0}{\Name_1}{{\Probe}}}\MNet \subst{\Name_0}{\MSer_0}\subst{\Name_1}{\MSer_1}\quad }&&
  \Rule[\rulescale]{\RuleMNetSTau}{\MNet(\Name) \trans{\tauof{\ProcAct}} \MSer }{\MNet \trans{\nettau{\Name}{\ProcAct}} \MNet \subst{\Name}{\MSer}}&&
  \Rule[\rulescale]{\RuleMNetTau}{\MNet(\Name) \trans{\mtauof{\MSerAct}} \MSer }{\MNet \trans{\mnettau{\Name}{\MSerAct}} \MNet \subst{\Name}{\MSer}}\quad \end{deduction}
 \caption{LTS semantics of monitored networks.}\label{fig:semantics-mnet}\end{figure}

\subsection{Instrumentation of Services and Networks, and Transparency}
\label{sec:monitoring:instr}
\label{sec:monitoring:transparency}

In \Cref{def:instr} below we formalise a class of functions that transform an
unmonitored network \(\Net\) into a monitored one, by instrumenting each service of \(\Net\) with a monitor.

\begin{definition}[Monitor instrumentation]
  \label{def:instr}A function \(\instr\) of type \(\NetType \to \MNetType\) is a \defn{monitor instrumentation}
  iff for any network \(\Net\) and name \(\Name \in \dom(\Net)\):
  \begin{enumerate}
  \item \(\Net(\Name) = \Ser\) implies \(\left(\instr(\Net)\right)(\Name) = \mqued{\MQue}{\MState}{\Ser}\) where
    \(\MQue\) is empty or only contains probes; and
  \item \(\left(\instr(\Net)\right)(\Name) = \mqued{\MQue}{\MState}{\Ser}\) implies \(\Net(\Name) = \Ser\) and
    \(\MQue\) is empty or only contains probes.
  \end{enumerate}
\end{definition}

\Cref{def:instr} requires that when \(\instr\) instruments a network \(\Net\),
the resulting monitored network does not contain any additional service messages
\(\Act\) in the monitor \changeOurs{}{style}{queues, since these} messages could disrupt the network.
Dually to instrumentation, in \Cref{def:deinstr} below we introduce a function to
\emph{deinstrument} a monitored network by ``stripping'' its monitors.

\begin{definition}[Monitor deinstrumentation]
  \label{def:deinstr}
  \label{def:strip}
  We define the \defn{deinstrumentation} \(\deinstr\) as the function of type \(\MNetType \to \NetType\)
  such that for any \(\MNet\) and \(\Name \in \dom(\MNet)\), we have \((\deinstr(\MNet))(\Name) = \strip{\MNet(\Name)}\), where the
  \defn{monitored service stripping operation \(\strip{\MSer}\)} is defined as follows:

  \smallskip \centerline{\(
    \begin{array}{rcl}
      \strip{\mquedp{\MQue}{\MState}{\IQue}{\Proc}{\OQue}}
      &=&\pqued{\queapp{\IQue}{\filterI{\MQue}}~}{~\Proc~}{~\queapp{\filterO{\MQue}}{\OQue}}
    \end{array}
  \)}\smallskip 

  \noindent where \(\filterI{\MQue}\) and \(\filterO{\MQue}\) are the longest subsequences
  from \(\MQue\) that only contain incoming and outgoing service messages, respectively.\note{Is it clear? The precise definition may not be more illuminating}\end{definition}

Intuitively, the deinstrumentation in \Cref{def:deinstr} strips all monitors and
monitor queues, and moves all incoming and outgoing service messages contained
in monitor queues into the input or output queue of the service. As a result,
the deinstrumented system can keep executing as if the monitors were never
instrumented; \change{change:retraction}{Mentioned that deinstrumentation is a retraction}{notably, deinstrumentation is the left-inverse (retraction) of any instrumentation}, as per \Cref{lem:strip-instr} below. Note that the \emph{vice versa} is not generally true: i.e., for arbitrary
instrumentation \(\instr\) and monitored network \(\MNet\) we may have \(\instr(\deinstr(\MNet)) \neq \MNet\),
because \(\MNet\) may contain probes that are stripped by \(\deinstr\) and instantiated differently by \(\instr\).

\begin{proposition}
  \label{lem:strip-instr}
  For any instrumentation \(\instr\), we have \(\deinstr(\instr(\Net)) = \Net\).
\end{proposition}

We can now prove that our generic black-box instrumentation/deinstrumentation in
\Cref{def:instr,def:deinstr} satisfy \Cref{crit:transparency}: they ensure
transparency for \emph{any} network, and also for any specific instance of
the monitoring algorithm \(\handle\), monitoring state \(\MState\), and probe
format \(\Probe\) (which we left unspecified in \Cref{def:syntax-mon}).
This result is shown in \Cref{lem:net-transp-compl,lem:net-transp-sound} below.

We first introduce some additional technical tools. In
\Cref{def:stripping-path} we extend the idea of ``monitor stripping'' by
removing all monitor actions from a monitored path \(\mpath\). 

\begin{definition}[Stripping of monitor actions]
  \label{def:stripping-path}
  The \defn{stripping of a monitored path \(\strip{\mpath}\)} is:
  \begin{equation*}
    \begin{array}{r@{\hspace{5pt}}c@{\hspace{5pt}}l@{\hspace{20pt}}r@{\hspace{5pt}}c@{\hspace{5pt}}l@{\hspace{20pt}}r@{\hspace{5pt}}c@{\hspace{5pt}}l}
      \strip{(\nettau{\Name}{\ProcAct} \pathseq \mpath)} &=& \nettau{\Name}{\ProcAct} \pathseq (\strip{\mpath})
      & \strip{(\mnettau{\Name}{\MSerAct} \pathseq \mpath)} &=& \strip{\mpath}
      & \strip{\pathnil} &=& \pathnil
      \\
      \strip{(\netcomm{\Name_0}{\Name_1}{\Chan} \pathseq \mpath)} &=& \netcomm{\Name_0}{\Name_1}{\Chan} \pathseq (\strip{\mpath})
      & \strip{(\mnetcomm{\Name_0}{\Name_1}{\Probe} \pathseq \mpath)} &=& \strip{\mpath}
    \end{array}
  \end{equation*}
\end{definition}

\begin{theorem}[Black-box instrumentation transparency --- completeness]
  \label{lem:net-transp-compl}For any instrumentation \(\instr\) and network \(\Net\), for any \(\Net'\) and path \(\ppath\),
  \(\Net \trans{\ppath} \Net'\) implies \(\exists \mpath',\MNet''\) such that \(\instr(\Net) \trans{\mpath'} \MNet''\) and \(\deinstr(\MNet'') = \Net'\).
\end{theorem}

\begin{proof}\renewcommand{\qedsymbol}{\coqed}
  We prove the statement as a consequence of the following stronger result:

  \medskip \hfill \begin{minipage}{0.8\linewidth}
    \begin{lemma}
      \label{lem:net-transp-compl:stronger}
      For any networks \(\Net\), \(\Net'\), path \(\ppath\), and instrumentation \(\instr\),
      we have that \(\Net \trans{\ppath} \Net'\) implies \(\exists \mpath',\instr'\) such that \(\strip{\mpath'} = \ppath\)
      and \(\instr(\Net) \trans{\mpath'} \instr'(\Net')\).\end{lemma}
  \end{minipage}
  \hfill \medskip 

  \noindent Then, from \Cref{lem:net-transp-compl:stronger} we obtain
  \Cref{lem:net-transp-compl} by constructing \(\mpath'\) from \(\ppath\) by
  adding suitable monitor actions (leading from \(\instr(\Net)\) to \(\instr'(\Net')\)), and
  then taking \(\MNet'' = \instr'(\MNet')\), which implies \(\deinstr(\MNet'') =
  \Net'\) (by \Cref{lem:strip-instr}).
\end{proof}

In order to prove the ``soundness'' part of \Cref{crit:transparency}
(\Cref{lem:net-transp-sound} below), in \Cref{def:flushing-path} we say that a
monitored path \(\mpath\) is \emph{monitor-flushing} if it only contains actions
that ``flush'' the monitor queues until they are empty or only contain probes,
thus leading to a state that corresponds to an instrumentation of some
unmonitored network (by \Cref{def:instr}). Note that a monitor-flushing path may
include the forwarding of incoming/outgoing service messages that are already in
the monitors' queues, but does \emph{not} include any execution step involving
the overseen services.

\begin{definition}[Monitor-flushing path]
  \label{def:flushing-path}
  We say that \defn{\(\mpath\) is a monitor-flushing path} if \(\mpath\) only contains actions of the form \(\mnettau{\Name}{\ProcAct}\), \(\mnetcomm{\Name_0}{\Name_1}{\Probe}\),
  or \(\netcomm{\Name_0}{\Name_1}{\Chan}\)
  (from \Cref{def:mon-network-semantics}).\end{definition}

\begin{theorem}[Black-box instrumentation transparency --- soundness]
  \label{lem:net-transp-sound}For any instrumentation \(\instr\) and network \(\Net\), for any monitored network \(\MNet'\) and monitored path \(\mpath\), \(\instr(\Net) \trans{\mpath} \MNet'\) implies \(\exists \MNet'', \Net''', \mpath', \ppath''\) such that \(\MNet' \trans{\mpath'} \MNet''\) and \(\deinstr(\MNet'') = \MNet'''\) and
  \(\Net \trans{\ppath''} \Net'''\).
\end{theorem}
\begin{proof}\renewcommand{\qedsymbol}{\coqed}
  We prove the statement as a consequence of the following stronger result:

  \medskip \hfill \begin{minipage}{0.8\linewidth}
    \begin{lemma}
      \label{lem:net-transp-sound:stronger}
      For any instrumentation \(\instr\) and network \(\Net\), for any monitored network \(\MNet'\) and monitored path \(\mpath\), \(\instr(\Net) \trans{\mpath} \MNet'\) implies \(\exists \Net''', \mpath', \instr'\) such that \(\mpath'\) is a flushing path, \(\MNet' \trans{\mpath'} \instr'(\Net''')\), and \(\Net \trans{\strip{(\mpath \pathseq \mpath')}} \Net'''\).
    \end{lemma}
  \end{minipage}
  \hfill \medskip 

  \noindent Then, from \Cref{lem:net-transp-sound:stronger} we obtain
  \Cref{lem:net-transp-sound} by 
  taking \(\Net'''\) and \(\MNet'' = \instr'(\Net''')\) (which implies
  \(\Net''' = \deinstr(\MNet'')\), by \Cref{lem:strip-instr}), constructing \(\mpath'\) as a flushing path from \(\MNet'\) to \(\MNet'' = \instr'(\Net''')\), and finally constructing \(\ppath'' = \strip{(\mpath \pathseq \mpath')}\)
  (i.e., stripping monitor actions by \Cref{def:stripping-path}).
\end{proof}

\note{There is more commented-out material here, but it does not seem necessary for the results above and in the next section. Double-check!}
 \section{A Distributed Black-Box Monitoring Algorithm for Deadlock Detection}\label{sec:algorithm}

In \Cref{def:deadlock-algo} below we introduce a deadlock detection
algorithm usable in our generic distributed black-box outline monitors
(\Cref{def:syntax-mon,def:mon-network,def:instr}), by suitably instantiating
their monitor state \(\MState\), the probes \(\Probe\), and the monitoring function \(\handle\).
The basic idea of our algorithm is inspired by~\cite{ChandyMisraHaas83,MitchellMerrit84}
\change{change:edge-chase}{Mention edge chasing earlier}{which estimate the dependency graph of (dead)locked processes via \emph{edge chasing}~\cite{MOSS1985,DBLP:journals/tse/SinhaN85}}: in a network \(\Net\), the monitor of an SRPC service named \(\Name\) aims at inferring whether
\(\Name\) is deadlocked by producing, forwarding, and receiving probes to/from other monitors. If a
monitor \(\MState\) receives a probe that it had \changeOurs{}{style}{itself produced and sent},
then (under some conditions) there is a dependency loop between locked services;
\changeOurs{}{clarity}{consequently, \(\MState\) can decidedly report} a deadlock.
The intuition behind this approach is supported by \Cref{lem:deadlock-cycle} below.

\begin{proposition}[Equivalence between deadlocks and lock-on cycles]
  \label{lem:deadlock-cycle}
  A service named \(\Name\) is deadlocked in \(\Net\) (by \Cref{def:deadlock})
  if and only if \(\Name\) is transitively locked on \changeOurs{}{ditto}{some} service \(\Name'\) that
  is transitively locked on itself (by \Cref{def:lock}).
\end{proposition}

\begin{definition}[Deadlock detection monitoring algorithm]
  \label{def:deadlock-algo}\label{def:deadlock-algo:probe}\label{def:deadlock-algo:state}\label{def:deadlock-algo:function}Let a \defn{monitor probe \(\Probe\)} be any value that a monitor can generate uniquely in
  the network \change{change:probe-unique}{Mention how unique probes are generated}{(e.g.~by combining its overseen service name with an incrementing counter)}. We will use the function \(\newprobe\) to instantiate a fresh probe.

  We define a \defn{monitor state \(\MState\)} as a record with the following fields:
  \begin{itemize}
  \item \field{probe} --- the probe associated with the current lock, or
    \(\noprobe\) if the monitor believes that the overseen service is not locked on any other service;
  \item \field{waiting} --- set of service names that the monitor considers locked on the overseen service;
  \item \field{alarm} --- boolean flag set to true when a deadlock is detected, and false otherwise.
  \end{itemize}

  For brevity, we will often write \(\netgetfield{\MNet}{\Name}{f}\) to refer to
  the field \field{f} of the monitor state \(\MState\) of the monitored service
  \(\MNet(\Name)\).
\label{def:active-probe}We say that a \defn{probe \(\Probe\) is active in the monitored network
  \(\MNet\)} if there is \(\Name\) such that \(\getfield{\MNet(\Name)}{probe} = \Probe \neq \noprobe\);
  in this case, we also say that \defn{\(\Name\) is the owner of the probe \(\Probe\)}.

  We define the \defn{monitor algorithm function \(\handle\)} as shown in \Cref{fig:algorithm}.
\end{definition}

\begin{figure}
  \centering
  \begin{minipage}[c]{0.68\linewidth}
    {\scriptsize \begin{align*}
  \handle\!\left(\MState, \recvq{\NameE}\right)
  &= \begin{cases}
    \left(\MState,\, \quenil\right)
    & \hspace{-5em}\text{if } \NameE = \bot \;\hfill \ruletag{\RuleAlgCIn}\\[0.5em]
    \left(\MState \withfield{waiting}{\getfield{\MState}{waiting} \cup \{\NameE\}},\, \quenil\right)
    & \hspace{-5em}\text{if } \getfield{\MState}{probe} = \noprobe \;\hfill \ruletag{\RuleAlgQInNotLocked}\\[0.5em]
    \left(\MState \withfield{waiting}{\getfield{\MState}{waiting} \cup \{\NameE\}},\, \msendr{\Name}{\getfield{\MState}{probe}}\right) & \text{else}
                                                                                                                             \;\hfill\ruletag{\RuleAlgQInLocked}\end{cases}\\
\handle\!\left(\MState, \sendq{\Name}\right)
  &= \left(\MState \withfield{probe}{\newprobe},\, \quenil\right)\;\hfill\ruletag{\RuleAlgQOut}\\
\handle\!\left(\MState, \recvr{\Name}\right)
  &= \left(\MState \withfield{probe}{\bot},\, \quenil\right)\;\hfill\ruletag{\RuleAlgRIn}\\
\handle\!\left(\MState, \sendr{\Name}\right)
  &= \left(\MState \withfield{waiting}{\getfield{\MState}{waiting} \setminus \{\Name\}},\, \quenil\right)\;\hfill\ruletag{\RuleAlgROut} \\
\handle\!\left(\MState, \mrecvr{\Name}{\Probe}\right)
  &=
    \begin{cases}
      \left(\MState \withfield{alarm}{\text{true}},\, \quenil\right) & \text{if } \getfield{\MState}{probe} = \Probe
                                                   \;\hfill\ruletag{\RuleAlgProbeInAlarm}\\[0.5em]
      \left(\MState, \quenil\right) & \text{if } \getfield{\MState}{probe} = \noprobe \;\hfill\ruletag{\RuleAlgProbeInNotLocked}\\[0.5em]
      \left(\MState,\, \left[\msendr{\Name'}{\Probe} \keyword{ for } \Name' \in \getfield{\MState}{waiting}\right]\right)
                                                 & \text{otherwise}\;\hfill\ruletag{\RuleAlgProbeInPropagate}
    \end{cases}\\
\handle\!\left(\MState, \sendc{\Name}\right)
  &= \left(\MState,\, \quenil\right)\;\hfill\ruletag{\RuleAlgCOut} \\
\end{align*}
 }
    \hspace{-27mm}
    \vspace{-3em}
    \captionof{figure}{The algorithm function \(\handle\) for our deadlock detection monitoring algorithm (\Cref{def:deadlock-algo}).}\label{fig:algorithm}
  \end{minipage}\hfill
  \begin{minipage}[c]{0.30\linewidth}
    \centering
    \scalebox{0.78}{\hspace{-5pt}\begin{tikzpicture}[->,>=stealth',auto,node distance=2.5cm, semithick,nam/.style={node distance=0.6cm},mon/.style={node distance=2.6cm},que/.style={node distance=0.75cm},ser/.style={node distance=0.75cm},every label/.style={align=left},every node/.style={align=center}]

\node (start) {};
  \node (init) at([shift={(-33pt,-5em)}]start) {};
  \node[mon] (m1) [right of=start, blue, node distance = 1em] {\(\strut\MState_1\)};
  \node[que] (q1) [left of=m1, blue] {\(\strut{\mquedl \MQue_1 \mquedsepl}\)};
  \node[nam] (n1) [left of=q1] {\(\strut{\Name_1}\!:\)};
  \node[ser] (s1) [right of=m1] {\(\strut{\mquedsepr \Ser_1 \mquedr}\)};

  \node[mon] (m3) at([shift={(1.3cm,1.7cm)}]m1) [blue] {\(\strut\MState_3\)};
  \node[que] (q3) [left of=m3, blue] {\(\strut{\mquedl \MQue_3 \mquedsepl}\)};
  \node[nam] (n3) [left of=q3] {\(\strut{\Name_3}\!:\)};
  \node[ser] (s3) [right of=m3] {\(\strut{\mquedsepr \Ser_3 \mquedr}\)};

  \node[mon] (m2) at([shift={(2.6cm,-0.5cm)}]m1) [blue] {\(\strut{\MState_2}\)};
  \node[que] (q2) [left of=m2, blue] {\(\strut{\mquedl \MQue_2 \mquedsepl}\)};
  \node[nam] (n2) [left of=q2] {\(\strut{\Name_2}\!:\)};
  \node[ser] (s2) [right of=m2]  {\(\strut{\mquedsepr \Ser_2 \mquedr}\)};

\path (init) edge[out=90,in=230,below right] node {1} (q1)
  (q1) edge[out=60, in=120, dotted] node[above] {\footnotesize{1.1}} (s1)
  (s1) edge[out=240, in=300, dotted] node[below] {\footnotesize{1.2}} (q1)
  (q1) edge[out=270, in=270] node[below] {2} (q2)
(q2) edge[out=60, in=120, dotted] node[above] {\footnotesize{2.1}} (s2)
  (s2) edge[out=240, in=300, dotted] node[below] {\footnotesize{2.2}} (q2)
  (n2) edge[out=90, in=270] node[near start, right=2pt] {3} (q3)
(q3) edge[out=60, in=120, dotted] node[above] {\footnotesize{3.1}} (s3)
  (s3) edge[out=240, in=300, dotted] node[below] {\footnotesize{3.2}} (q3)
  (n3) edge[out=190, in=90] node[above left] {4} (q1)
(m1) edge[out=120, in=210, red, dashed] node {\textbf{5}} (q3)
  (m3) edge[bend left, red, dashed] node {\textbf{6}} (q2)
  (m2) edge[bend left, red, dashed] node {\textbf{7}} (q1)
  ;
\end{tikzpicture}

 \hspace{-6pt}}
\vspace{-2mm}\captionof{figure}{Steps of a deadlock detection from \Cref{ex:deadlock-algo}.
      Numbers indicate the order of events; black solid arrows represent
      queries, red dashed arrows represent probes, dotted arrows represent
      monitors forwarding/intercepting queries to/from the services they oversee.}\label{fig:example-detection}\end{minipage}
\end{figure}

Recall that, by the monitored service semantics (\cref{fig:semantics-mon},
\Cref{\RuleMonTauIn,\RuleMonTauMonIn,\RuleMonOut}), the monitor algorithm
function \(\handle\) is \changeOurs{}{style}{tasked with processing} every incoming/outgoing service
message and every incoming probe at the front of the monitor queue. 
\changeOurs{}{ditto}{Following this,} \(\handle\) returns an updated monitor state and a (possibly empty) sequence of
outgoing probes towards other monitors; such outgoing probes are added at the
front of the monitor queue, and then sent to their recipients (by
\Cref{\RuleMonMonOut} in \Cref{fig:semantics-mon}). By \Cref{fig:algorithm}, our
algorithm behaves as follows, with \(\MState\) being the monitor state, and
\(\Ser\) the service overseen by the monitor:

\begin{itemize}
\item If \(\handle\) sees an incoming cast \(\recvc\), then it ignores it by \cref{\RuleAlgCIn}.
\item If \(\handle\) sees an incoming query \(\recvq{\Name}\), then
  \cref{\RuleAlgQInNotLocked,\RuleAlgQInLocked} add \(\Name\) to the
  \field{waiting} set of \(\MState\). Moreover, if
  \(\getfield{\MState}{probe}\) is not \(\noprobe\) (i.e., \(\Ser\) appears
  locked on some service), then \cref{\RuleAlgQInNotLocked} sends
  \(\getfield{\MState}{probe}\) to \(\Name\)'s monitor.
\item If \(\handle\) sees an outgoing query \(\sendq{\Name}\), then
  \cref{\RuleAlgQOut} sets \(\getfield{\MState}{probe}\) to a new probe
  \(\Probe\) (meaning that \(\Ser\) appears locked on some service); note that
  the new probe \(\Probe\) is \emph{fresh}, i.e., globally unique in the
  monitored network.
\item If \(\handle\) sees an incoming response \(\recvr{\Name}\), then
  \cref{\RuleAlgRIn} sets \(\getfield{\MState}{probe}\) to
  \(\noprobe\) (meaning that \(\Ser\) appears now unlocked).
\item If \(\handle\) sees an outgoing response \(\sendr{\Name}\), then
  \cref{\RuleAlgROut} removes the recipient \(\Name\) from the
  \field{waiting} set of \(\MState\).
\item If \(\handle\) sees an incoming probe \(\mrecvr{\Name}{\Probe}\), then:
  \begin{itemize}[leftmargin=*]
  \item if \(\Probe\) was previously produced by \(\MState\) (because
    \(\getfield{\MState}{probe} = \Probe\), hence \(\MState\) owns the probe),
    then \cref{\RuleAlgProbeInAlarm} reports a deadlock by setting
    \(\getfield{\MState}{alarm}\) to true;
  \item otherwise, if \(\Ser\) does not seem locked (since
    \(\getfield{\MState}{probe} = \noprobe\)), then
    \cref{\RuleAlgProbeInNotLocked} ignores the probe;
  \item otherwise (i.e., \changeCameraReady{Clarify}{when \(\Ser\) appears locked
      (\(\getfield{\MState}{probe} \neq \noprobe\)), but \(\Ser\) does not own the probe because
      \(\getfield{\MState}{probe} \neq \Probe\)}), \cref{\RuleAlgProbeInPropagate} forwards the
    probe \(\Probe\) to every \(\Name'\) in the \field{waiting} set of \(\MState\).
  \end{itemize}
\item If \(\handle\) sees an outgoing cast \(\sendc{\Name}\), then \(\handle\) ignores it by \cref{\RuleAlgCOut}.
\end{itemize}

\begin{example}[Deadlock detection algorithm in action]
  \label{ex:deadlock-algo}
  \Cref{fig:example-detection} depicts the detection of a deadlock involving
  three monitored services. In order to serve an initial query (arrow number 1),
  the services named \(\Name_1\), \(\Name_2\) and \(\Name_3\) end up querying
  each other (arrows 2, 3, and 4) forming a locked-on cycle. 

  \begin{itemize}
  \item When \(\Name_1\) sends a query to \(\Name_2\) (arrow number 2), the algorithm \(\handle\)
    updates the state \(\MState_1\) setting \(\getfield{\MState_1}{probe} = \Probe \neq \noprobe\)
    (by \cref{\RuleAlgQOut} in \Cref{fig:algorithm}, meaning the monitor believes \(\Ser_1\) is
    locked).
  \item When \(\Name_2\) and \(\Name_3\) receive their queries (arrows 2 and 3), \(\handle\) updates their
    \field{waiting} sets to \(\{\Name_1\}\) and \(\{\Name_2\}\), respectively (by
    \cref{\RuleAlgQInNotLocked}).
  \item When \(\Name_2\) and \(\Name_3\) send their queries (arrows 3 and 4), \(\handle\) sets their
    respective \field{probe}s to values different from \(\noprobe\) (\cref{\RuleAlgQOut}; the
    monitors believe \(\Ser_2\) and \(\Ser_3\) are locked).
  \item After \(\handle\) running in \(\Name_1\) sees an incoming query \(\recvq{\Name_3}\) (arrow
    number 4), it sends \(\getfield{\MState_1}{probe} = \Probe \neq \noprobe\) to the monitor of
    \(\Name_3\) (arrow number 5, by \cref{\RuleAlgQInLocked}).
  \item Then, the monitors of \(\Name_3\) and \(\Name_2\) forward \(\Probe\) to the services in their
    \field{waiting} sets (arrows 6 and 7, by \cref{\RuleAlgProbeInPropagate}).
  \item Finally, the probe \(\Probe\) is received by the monitor of \(\Name_1\), which still has
    \(\getfield{\MState_1}{probe} = \Probe\). Therefore, the function \(\handle\) reports a deadlock (by \cref{\RuleAlgProbeInAlarm}).
  \end{itemize}
\end{example}

 \section{Proving the Preciseness of Our Distributed Deadlock Detection Monitors}\label{sec:proof}

At this stage, \Cref{lem:net-transp-compl,lem:net-transp-sound} ensure that any monitor
instrumentation \(\instr\) is transparent for any network \(\Net\), thus achieving
\Cref{crit:transparency} of monitoring correctness. We now also prove that, if \(\instr\) uses our deadlock detection algorithm
(\Cref{def:deadlock-algo}) and \(\Net\) consists of SRPC services (\Cref{def:srpc-service}), then
\(\instr\) performs sound and complete deadlock detection for \(\Net\), thus achieving
\Cref{crit:detection-preciseness} of monitoring correctness. We present this result in \Cref{sec:proof:detect-precise}
(\Cref{thm:deadlock-detection-preciseness}); to obtain it, we identify and prove
detailed invariants of SRPC networks (\Cref{sec:proof:invariants-srpc}) and monitor states
w.r.t.~the surrounding network
(\Cref{sec:proof:invariants-mon-sound} and \Cref{sec:proof:invariants-mon-compl}).
\change{change:proof:invariants-hard}{Highlight challenge in finding invariants}{Identifying, formalising, and proving such invariants was a major
  challenge of this work, and this is the stage where potential and subtle counterexamples
  to correct deadlock monitoring may be found (\Cref{sec:related:edge} discusses one such cases, for a different monitoring algorithm).}

\subsection{\changeNoMargin{Well-Formedness as an Invariant of SRPC Networks}}
\label{sec:proof:invariants-srpc}

\change{change:proof:wf-srpc-title}{Revise section title for clarity}{}

\change{change:proof:wf-overview}{Move paragraph earlier and clarify}{In this section we formalise a series of \emph{well-formedness} properties for
  SRPC services and networks
  (\Cref{def:well-formed-srpc-serv,def:well-formed-srpc-net} below). Such properties ensure the absence of issues such as: service \(\Name\) sending a response to service \(\Name'\),
  although \(\Name'\) had not previously sent a query to \(\Name\). Crucially, in \Cref{lem:well-formedness-persistent} below we prove that such
  properties are preserved throughout any execution from the moment any
  \emph{arbitrary} SRPC network is initially deployed.
To this end, }in \Cref{def:initial-net}, we formalise a network  \changeOurs{}{style}{in an
initial state (i.e., just deployed)}, and a corresponding monitor instrumentation that is not (yet)
reporting any deadlock while assuming that all services are unlocked.

\begin{definition}[Initial network and instrumentation]
  \label{def:initial-net}\label{def:initial-instr}We say that \defn{\(\Net\) is an initial network} when:

  \begin{enumerate}
\item the input and output queues of each service are empty;
  \item no service is locked on any other service;
  \item \label{item:init-state-busy} for any service in the \(\srpcw{\NameE}\) state, \(\NameE = \nameext\).
  \end{enumerate}

  The \defn{initial deadlock detection monitor instrumentation \(\instr\) for
  \(\Net\)} assigns a monitor to each service in \(\Net\) using the deadlock
  detection algorithm \(\handle\) in \Cref{def:deadlock-algo}, with equal
  monitor states \(\MState\) such that \(\getfield{\MState}{probe} = \noprobe\),
  \(\getfield{\MState}{waiting} = \emptyset\), and \(\getfield{\MState}{alarm}\)
  false.\end{definition}

In the item~\ref{item:init-state-busy} of \Cref{def:initial-net}, we allow services to be initialised
as \(\srpcw{\nameext}\) to model client applications that engage with the network. Such services can
immediately send queries and casts to other services, and thus initiate traffic in the network.
\change{change:proof:client-clarify}{Slight rewording for clarity}{To formalise the SRPC well-formedness invariant we also need the notion}
of \emph{client} (\Cref{def:client}),
i.e., a service that is being served (or is about to be served) by another
service.

\begin{definition}[Client]
  \label{def:client}We say \defn{\(\Name\) is a client of SRPC service \(\Ser\)} \!\(\vphantom{x} = \pqued{\IQue}{\Proc}{\OQue}\) if \emph{either}:
  \begin{itemize}
  \item \(\IQue \trans{\qpopq{\Name}}\) --- i.e., there is an incoming query from \(\Name\); or
  \item \(\issrpcw{\Proc}{\Name}\) or \(\issrpcl{\Proc}{\Name}{\_}\) --- i.e., the
    process is already handling a query of \(\Name\); or
  \item \(\OQue \trans{\qpopr{\Name}}\) --- i.e., there is an outgoing response to \(\Name\).
  \end{itemize}
\end{definition}

By \Cref{def:client}, if service \(\Name'\) is ``in between'' receiving a query from \(\Name\) and sending a response
to \(\Name\), then \(\Name\) is a client of \(\Name'\).
\change{change:proof:wf-intro}{Add intuition of SRPC well-formedness}{We now have all the ingredients to formalise the SRPC well-formedness of services and networks, as per \Cref{def:well-formed-srpc-serv,def:well-formed-srpc-net} below. 
  \noindent Intuitively, SRPC well-formedness means that:

  \begin{itemize}
  \item At all times, every service has at most one query being remotely handled;
  \item Each response is always associated to a past query;
  \item SRPC states remain aligned to the queries and responses in the network. For example, if a query
    is queued somewhere, then its sender is in SRPC \(\srpclname\) state.
  \end{itemize}
}

\begin{definition}[SRPC well-formedness]
  \label{def:well-formed-srpc-serv}Consider the service \(\Ser = \pqued{\IQue}{\Proc}{\OQue}\). We say that
  \defn{\(\Ser\) is well-formed as an SRPC client} if \emph{all} the following
  conditions hold:

  \begin{enumerate}
  \item in \(\IQue\) there is at most one response;
  \item if \(\IQue\) contains a response from \(\Name_s\), then \(\issrpcl{\Proc}{\Name_c}{\Name_s}\);
  \item if \(\OQue\) contains a query to \(\Name_s\), then \(\issrpcl{\Proc}{\Name_c}{\Name_s}\);
  \item if \(\OQue = \OQue_l \queappsym \queelq{\Name_s} \queappsym \OQue_r\), then \(\OQue_r = \quenil\);
  \item if there is a response in \(\IQue\), then there is no query in \(\OQue\);
  \item if there is a query in \(\OQue\), then there is no response in \(\IQue\);
  \item if \(\issrpcl{\Proc}{\Name_c}{\Name_s}\), then \(\OQue\) is either empty or contains \(\queelq{\Name_s}\).
  \end{enumerate}

  \noindent We say that \defn{\(\Ser\) is well-formed as an SRPC server} if \emph{all} the following conditions
  hold:

  \begin{enumerate}
  \item in \(\IQue\) there is at most one query per service name;
  \item in \(\OQue\) there is at most one response per service name;
  \item\label{it:wf:excl} for every service \(\Name_c\), the following are mutually exclusive: (A) in \(\IQue\) there is a query from \(\Name_c\); or (B) \(\issrpcw{\Proc}{\Name}\); or (C) \(\issrpcl{\Proc}{\Name}{\_}\); or (D) in \(\OQue\) there is a response to \(\Name_c\).
  \end{enumerate}

 \noindent We say \defn{\(\Ser\) is a well-formed SRPC service} if \(\Ser\) is
  well-formed as both SRPC client and server.
\end{definition}

By \Cref{def:well-formed-srpc-serv} an SRPC service \(\Ser\) is well-formed as a client
if its SRPC state reflects the contents of its queues and \emph{vice versa}, and
if \(\Ser\) is free from anomalies such as receiving a response before a query is sent.
Well-formedness as a server guarantees that \(\Ser\) has not received a
query from a service \(\Name\) without replying to all \(\Name\)'s
previous queries. \Cref{def:well-formed-srpc-net} extends well-formedness
to networks.

\begin{definition}
  \label{def:well-formed-srpc-net}\label{def:well-formed-srpc-net-mon}We say that \defn{\(\Net\) is a well-formed SRPC network} when:
  \begin{enumerate}
  \item for every \(\Name \in \dom(\Net)\), \(\Net(\Name)\) is a well-formed SRPC service (by
    \Cref{def:well-formed-srpc-serv});
  \item \(\Name\) is a client of \(\Name'\) (by \Cref{def:client}) if and only if \(\Name\) is locked on \(\Name'\).
  \end{enumerate}
  By extension, we say that \defn{\(\MNet\) is a well-formed monitored SRPC
  network} if \(\deinstr(\MNet)\) (from \Cref{def:deinstr}) is a well-formed SRPC
  network.
\end{definition}

\change{change:proof:wf-result-example}{Added for better flow}{We conclude this section by proving that \emph{any} initial SRPC network (with
  or without monitors) is well-formed, and preserves its well-formedness at
  run-time (\Cref{lem:well-formedness-persistent}). We also show an example of
  how well-formedness is preserved throughout an execution
  (\Cref{ex:srpc-well-formed-exec}).}

\begin{theorem}[Well-formedness is an invariant]
  \label{lem:well-formedness-persistent}If \(\Net\) is an initial SRPC network and \(\Net \trans{\ppath} \Net'\),
  then \(\Net'\) is a well-formed SRPC network; moreover, for any instrumentation \(\instr\), if \(\instr(\Net) \trans{\mpath'}
  \MNet''\), then \(\MNet''\) is a well-formed monitored SRPC
  network.
\end{theorem}
\begin{proof}\renewcommand{\qedsymbol}{\coqed}
The fact that any initial SRPC network \(\Net\) is well-formed follows by
  \Cref{def:initial-net,def:well-formed-srpc-net}. 
We prove that if any \(\Net''\) is SRPC and well-formed and \(\Net'' \trans{\Act} \Net'''\),
  then \(\Net'''\) is also SRPC and well-formed, by case analysis of \(\Act\). Finally, we prove the two parts of the statement by induction on \changeOurs{}{clarity}{the structure of} \(\ppath\) and \(\mpath'\), and
  using \Cref{def:deinstr}.
\end{proof}

\begin{figure}
  \centering {\footnotesize \begin{tabular}{|l|c|c|c|}
  \hline
  Name & Initial SRPC state & Intermediate state & Once deadlocked \\
  \hline
  \(E_1\) & \(\pqued{\quenil}{\srpcw{\nameext}}{\quenil}\)
       & \(\pqued {\quenil}{\srpcl{\nameext}{P_2}}{\queelq{P_2}}\)
                            & \(\pqued {\queelq{P_1}}{\srpcl{\nameext}{P_2}}{\quenil}\) \\
  \hline
  \(E_2\) & \(\pqued{\quenil}{\srpcw{\nameext}}{\quenil}\)
       & \(\pqued {\quenil}{\srpcl{\nameext}{P_3}}{\quenil}\)
                            & \(\pqued {\queelq{P_2}}{\srpcl{\nameext}{P_3}}{\quenil}\) \\
  \hline
  \(E_3\) & \(\pqued{\quenil}{\srpcw{\nameext}}{\quenil}\)
       & \(\pqued {\quenil}{\srpcl{\nameext}{P_1}}{\quenil}\)
                            & \(\pqued {\queelq{P_3}}{\srpcl{\nameext}{P_1}}{\quenil}\) \\
  \hline
  \(P_1\) & \(\pqued{\quenil}{\srpcf}{\quenil}\)
       & \(\pqued {\queelq{E_3}}{\srpcf}{\quenil}\)
                            & \(\pqued {\quenil}{\srpcl{E_3}{E_1}}{\quenil}\) \\
  \hline
  \(P_2\) & \(\pqued{\quenil}{\srpcf}{\quenil}\)
       & \(\pqued {\quenil}{\srpcf}{\quenil}\)
                            & \(\pqued {\quenil}{\srpcl{E_1}{E_2}}{\quenil}\) \\
  \hline
  \(P_3\) & \(\pqued{\quenil}{\srpcf}{\quenil}\)
       & \(\pqued {\quenil}{\srpcw{E_2}}{\quenil}\)
                            & \(\pqued {\quenil}{\srpcl{E_2}{E_3}}{\quenil}\) \\
  \hline

\end{tabular}
 }
\vspace{-2mm}\caption{Some execution states of the network from \Cref{fig:deadlock-envelope}.
   All these configurations are well-formed.}\label{fig:well-formed-envelope}\end{figure}

\begin{example}[Execution of a well-formed SRPC network]
  \label{ex:srpc-well-formed-exec}\Cref{fig:well-formed-envelope} lists the queues and SRPC states of each
  service of the network in \Cref{fig:deadlock-envelope}, in three stages of
  execution: initial, intermediate, and deadlocked. All states are well-formed,
  by \Cref{def:well-formed-srpc-serv,def:well-formed-srpc-net}. Initially,
  \(E_1\), \(E_2\) and \(E_3\) are \(\srpcw{\nameext}\), while \(P_1\), \(P_2\)
  and \(P_3\) are \(\srpcf\) with all queues empty, trivially fulfilling all
  requirements for well-formedness. The second column shows an intermediate
  state after queries marked in the picture with 1,2,3,5,6 were sent: \(E_1\) is
  about to send a query to \(P_2\), meanwhile \(E_2\) and \(E_3\) are already
  locked on \(P_3\) and \(P_1\) respectively. \(P_1\) has just received the
  query from \(E_3\) into its input queue, while \(P_3\) is already processing
  the query from \(E_2\). Finally in the third column, all queries have been
  delivered and a deadlock occurred. Note that parameters of SRPC states always
  correspond to either a query in one of the queues or a parameter of the SRPC
  state of another service.
\end{example}

\subsection{Monitor Knowledge Invariant for Complete Deadlock Detection}
\label{sec:proof:invariants-mon-compl}

To ensure that our algorithm detects all deadlocks, the monitor states must not underestimate when the services they oversee are
locked. 
This is crucial when the monitors receive probes, as ignoring an active probe may cause a missing deadlock report (i.e., a false
negative). \Cref{def:kic} below formalises when monitors have complete
(and possibly over-approximated) knowledge of the network locks, via a
series of implications from the state of the network to the states of the
monitors, ensuring that each deadlock is \changeOursNoMargin{either} already reported or about to be
reported (by \Cref{def:alarm-condition}).

\begin{definition}[Complete lock knowledge]
  \label{def:kic}
  We say that \defn{monitors have complete lock knowledge in \(\MNet\)} if,
  for any \(\Name,\Name' \in \dom(\MNet)\),
  \changeOurs{}{clarity}{\emph{all} of the following conditions hold:}
  \begin{enumerate}
  \item if \(\Name\) is locked in \(\deinstr(\MNet)\), then \(\netgetfield{\MNet}{\Name}{probe} \neq \noprobe\);
  \item if \(\Name\) is locked on \(\Name'\) in \(\deinstr(\MNet)\), then either \(\Name \in \netgetfield{\MNet}{\Name'}{waiting}\) or
    there is an incoming query from \(\Name\) in the monitor queue of \(\Name'\);
  \item if \(\Name\) is transitively locked on itself in \(\deinstr(\MNet)\), then \(\Name\) is
    transitively locked on a service with the \emph{alarm condition}
    (\Cref{def:alarm-condition} below).
  \end{enumerate}
\end{definition}

\begin{definition}[Alarm condition]
  \label{def:alarm-condition}We say that \defn{\(\Name\) has the alarm condition in \(\MNet\)} when \emph{all} of the following conditions hold:
  \begin{enumerate}[label=(\alph*)]
  \item \(\Name\) is transitively locked on some \(\Name_1\) in \(\deinstr(\MNet)\) (which may be \(\Name\) itself);
  \item \(\Name_1\) is locked on some \(\Name_2\) that is also locked in \(\deinstr(\MNet)\) (where
    \(\Name_1,\Name_2\) may coincide);\item\label{item:ac:probes-def} monitors of \(\Name_2\) and all services transitively locked on
    \(\Name_2\) in \(\deinstr(\MNet)\) have their \(\field{probe} \neq \noprobe\);
  \item\label{item:ac:mon-que} \emph{any} of the following conditions holds:
    \begin{enumerate}[label=(\arabic*)]
\item\label{item:ac:alarm} \(\netgetfield{\MNet}{\Name}{alarm} = \text{true}\);
    \item\label{item:ac:fin} The monitor queue of \(\Name\) contains an incoming active probe of
      \(\Name\);
    \item\label{item:ac:send} The monitor queue of \(\Name_2\) contains an outgoing active probe of
      \(\Name\) towards \(\Name_1\);
    \item\label{item:ac:wait} The monitor queue of \(\Name_2\) contains an incoming active probe of
      \(\Name\), and \(\Name_1~\in~\netgetfield{\MNet}{\Name_2}{waiting}\);
    \item\label{item:ac:query} The monitor queue of \(\Name_2\) contains an incoming query from \(\Name_1\)
      followed by an incoming active probe of \(\Name\);
    \item\label{item:ac:close} The monitor queue of \(\Name\) contains an incoming query from \(\Name_1\).
    \end{enumerate}
  \end{enumerate}
\end{definition}

Intuitively, \Cref{def:alarm-condition} holds in a monitored network when an alarm is already reported
(item \ref{item:ac:alarm}) or is inevitable, because the monitors of locked services are ready to
propagate probes (by item~\ref{item:ac:probes-def} and \cref{\RuleAlgProbeInPropagate,\RuleAlgQInLocked} in
\Cref{fig:algorithm}) and the monitor queues' configuration will allow the propagation (by
items~\ref{item:ac:fin} to \ref{item:ac:close}).
This, together with well-formedness (\Cref{def:well-formed-srpc-net}),
guarantees that the active probe owned by the monitor of a deadlocked \(\Name\)
will remain active and will be propagated by all monitors on its way to
\(\Name\). This is proven in \Cref{thm:alarm-condition} below.

\begin{lemma}[Alarm condition leads to alarm]
  \label{thm:alarm-condition}
  If \(\MNet\) is a well-formed monitored SRPC network and \(\Name\) has the alarm
  condition in \(\MNet\) (by \Cref{def:alarm-condition}), then \(\exists \mpath,\MNet'\) such that \(\mpath\) is a monitor-flushing path (by
  \Cref{def:flushing-path}) and \(\MNet \trans{\mpath} \MNet'\) and
  \(\netgetfield{\MNet'}{\Name}{alarm} = \text{true}\).
\end{lemma}

\begin{theorem}[Complete lock knowledge is an invariant]
  \label{lem:kic-persistent}If \(\Net\) is an initial SRPC network, \(\instr\) is the initial deadlock
  detection monitor instrumentation for \(\Net\), and \(\instr(\Net)
  \trans{\mpath} \MNet'\), then the monitors in \(\MNet'\) have complete lock
  knowledge.
\end{theorem}
\begin{proof}\renewcommand{\qedsymbol}{\coqed}
We first easily prove that for any initial SRPC network \(\Net\), the monitors
  in \(\instr(\Net)\) (with \(\instr\) being initial for \(\Net)\) have complete
  lock knowledge, by \Cref{def:initial-net,def:kic}. Then, we prove an intermediate result: if \(\MNet''\) is SRPC-well-formed and
  has complete lock knowledge and \(\MNet'' \trans{\MAct} \MNet'''\), then
  \(\MNet'''\) is also SRPC-well-formed and still has complete lock knowledge. Finally, we prove the main statement by induction on \(\mpath\), using
  \Cref{lem:well-formedness-persistent} and the intermediate result above.\end{proof}

\subsection{Monitor Knowledge Invariants for Sound Deadlock Detection}
\label{sec:proof:invariants-mon-sound}

We now address the dual problem w.r.t.~\Cref{sec:proof:invariants-mon-compl}:
to ensure that our algorithm only detects real deadlocks, the monitor states must not overestimate whether the services they oversee are
locked, thus avoiding false deadlock reports (i.e., false positives). \Cref{def:kis} below formalises when monitors have sound (and
possibly under-approximated) knowledge of the network locks, through a series of
implications from the states of the monitors to the state of the network,
ensuring that each alarm (already raised or about to be raised) corresponds to a
real deadlock. We then prove that \Cref{def:kis} is an invariant (in
\Cref{lem:kis-persistent} below).

\begin{definition}[Sound lock knowledge]
  \label{def:kis}
  We say that \defn{monitors have sound lock knowledge in \(\MNet\)} if,
  for any \(\Name,\Name' \in \dom(\MNet)\),
  \changeOurs{}{clarity}{\emph{all} of the following conditions hold:}

  \begin{enumerate}
  \item if \(\netgetfield{\MNet}{\Name}{probe}\!\neq\!\noprobe\) then \(\Name\) is locked in \(\deinstr(\MNet)\) or its monitor queue has an incoming response;
  \item if \(\Name \in \netgetfield{\MNet}{\Name'}{waiting}\), then \(\Name\) is locked on \(\Name'\) in \(\deinstr(\MNet)\);
  \item if \(\Name'\) has an outgoing probe to \(\Name\) in its monitor queue, then \(\Name\) is locked on \(\Name'\) in \(\deinstr(\MNet)\);
  \item if \(\Name\) has the probe \(\netgetfield{\MNet}{\Name'}{probe}\) (either outgoing or incoming) in its monitor queue, then \(\Name\) is transitively
    locked on \(\Name'\) in \(\deinstr(\MNet)\);
  \item\label{item:kis:alarm} if \(\netgetfield{\MNet}{\Name}{alarm} = \text{true}\), then \(\Name\) is transitively locked on itself in \(\deinstr(\MNet)\).
  \end{enumerate}
\end{definition}

\begin{theorem}[Sound lock knowledge is an invariant]
  \label{lem:kis-persistent}If \(\Net\) is an initial SRPC network, \(\instr\) is the initial deadlock
  detection monitor instrumentation for \(\Net\), and \(\instr(\Net)
  \trans{\mpath} \MNet'\), then the monitors in \(\MNet'\) have sound lock
  knowledge.
\end{theorem}
\begin{proof}\renewcommand{\qedsymbol}{\coqed}
  The proof has the same structure of the proof of \Cref{lem:kic-persistent},
  but replacing complete lock knowledge (\Cref{def:kic}) with sound lock
  knowledge (\Cref{def:kis}).
\end{proof}

\subsection{Deadlock Detection Preciseness Result}
\label{sec:proof:detect-precise}

We now finally prove that in any initial SRPC network, our
black-box monitors (\Cref{def:instr,def:mon-network}) with our
deadlock detection algorithm (\Cref{def:deadlock-algo,def:initial-instr}) perform
sound and complete deadlock detection, achieving the correctness
\Cref{crit:detection-preciseness} (in addition to \Cref{crit:transparency},
proven in \Cref{sec:monitoring:transparency}).

\begin{theorem}[Deadlock detection preciseness]
  \label{thm:deadlock-detection-preciseness}
  Let \(\Net\) be an initial SRPC network, and \(\instr\) be the initial deadlock detection instrumentation for \(\Net\).
  Then, for any monitored network \(\MNet'\) and monitored path \(\mpath\) such that \(\instr(\Net) \trans{\mpath} \MNet'\):
  \begin{description}
  \item[Completeness:]If there is a deadlock in $\deinstr(\MNet')$, then there are $\MNet'', \mpath'$ such that
      \(\MNet' \trans{\mpath'} \MNet''\) and $\MNet''$ reports a deadlock.\item[Soundness:]If $\MNet'$ reports a deadlock, then there is a deadlock in $\deinstr(\MNet')$.
  \end{description}
\end{theorem}
\begin{proof}\renewcommand{\qedsymbol}{\coqed}
  Observe that \(\MNet'\) is a well-formed monitored SRPC network (by \Cref{lem:well-formedness-persistent})
  and its monitors have sound and complete lock knowledge (by \Cref{lem:kis-persistent,lem:kic-persistent}).
  We prove the ``completeness'' result using \Cref{def:kic}, \Cref{lem:deadlock-cycle}, and
  \Cref{thm:alarm-condition}, which constructs \(\mpath'\) as a monitor-flushing
  path going from \(\MNet'\) to a monitored network configuration \(\MNet''\) which reports a
  deadlock. We prove the ``soundness'' result by item~\ref{item:kis:alarm} of \Cref{def:kis}.\end{proof}

Notice that the ``completeness'' of \Cref{crit:detection-preciseness} and
\Cref{thm:deadlock-detection-preciseness} above does not guarantee that every
deadlock \emph{will} be reported: it only ensures that after a deadlock occurs,
\emph{it is always possible} for the monitors to raise an alarm in every
\(\MNet'\) reached by the monitored network; for this to happen, \(\MNet'\) must
follow a monitor-flushing path $\mpath'$ constructed by
\Cref{thm:alarm-condition}. This, in principle, leaves to \(\MNet'\) the
possibility of \emph{never} following that path. However, in \Cref{thm:detect-completeness-upgrade} below we show that,
assuming \emph{fairness of components}~\cite{Glabbeek_2019},
the deadlock \emph{will} be eventually reported.

\begin{theorem}[Eventual deadlock reporting]\label{thm:detect-completeness-upgrade}
Assume that the premises of \Cref{thm:deadlock-detection-preciseness} hold and
  there is a deadlock in $\deinstr(\MNet')$. Also assume that \(\MNet'\) is
  executed with \emph{fairness of components}~\cite{Glabbeek_2019}. Then, the deadlock
  will be eventually reported.
\end{theorem}

\change{change:strong-compl-proof}{Added proof outline of eventual deadlock reporting}{\begin{proof}\renewcommand{\qedsymbol}{\coqed}
  From the monitor-flushing path constructed by \Cref{thm:alarm-condition} in a deadlocked monitored network (such as $\MNet'$), we define a measure \((m, k)\) where \(m\) is the number of times
  an active probe \(\Probe\) must be forwarded in order to reach its owner
  (i.e., the monitor that created \(\Probe\)), and \(k\) is
  the number of flushing transitions (\Cref{def:flushing-path}) needed to forward \(\Probe\) to the next monitor. We compare measures lexicographically, e.g.~\((1, 9) < (2, 1)\).
We prove that this measure never increases as $\MNet'$ runs; moreover, in every $\MNet''$ reachable from $\MNet'$ 
  there is 
a monitor whose \emph{only} enabled transition decreases the measure
(unless it is already \((0, 0)\)) --- specifically, the monitor with \(\Probe\)
  in its queue reduces the measure by performing a flushing transition. Finally, we prove that if the measure reaches \((0, 0)\), the
  monitor that owns \(\Probe\) raises an alarm.
Assuming fairness of components~\cite{Glabbeek_2019},
  every monitor that reduces the current measure will be eventually scheduled to
  perform its transitions --- hence, the deadlock will be eventually reported. \emph{(Note: this last sentence from ``Assuming\ldots'' is not part of our Coq mechanisation.)}
For more details, see~\iftoggle{techreport}{\Cref{app:proof-completeness-upgraded}}{\cite[Appendix C]{TECHREPORT}}.\end{proof}
}
 \section{\Toolname, a Prototype Monitoring Tool for Distributed Deadlock Detection}\label{sec:tool}

We now present \emph{\toolname}, a proof-of-concept monitoring tool that
realises the distributed black-box instrumentation modelled in \Cref{sec:monitoring} with
our algorithm in \Cref{def:deadlock-algo}. \Toolname instruments monitors in Erlang and Elixir programs based on generic
servers (\texttt{gen\_server}), with minimal user intervention. 
We describe the \toolname usage and implementation in \Cref{sec:tool:implementation},
and evaluate its performance in \Cref{sec:tool:evaluation}.
\Toolname is the companion artifact of this work~\cite{rowicki_2025_16909304} and its latest version is
available at \url{https://github.com/radrow/ddmon}.
We provide additional details about \Toolname in \iftoggle{techreport}{\Cref{app:implementation}}{\cite[Appendix A]{TECHREPORT}}.

\subsection{Usage and Implementation}\label{sec:tool:implementation}

\Toolname acts as a drop-in replacement for the \texttt{gen\_server} behaviour
in the OTP standard library \cite{erlang-doc-gen-server}. The intended use case is that \Toolname can be activated during the testing or
staging of a \texttt{gen\_server}-based application, or even in production
(although the monitoring overhead evaluated in \Cref{sec:tool:evaluation} must
be taken into account). An Elixir programmer can deploy \Toolname monitors in a
\texttt{gen\_server}-based application by defining the following alias, without
changing other code:\footnote{\label{footnote:alternative-instrumentation}To remove the need for the alias, we could instead implement
  \Toolname by extending the system-wide \texttt{gen\_server} behaviour to
  (optionally) include our monitors. Alternatively, \Toolname could perform dynamic code replacement to inject
  monitors in a running system. The approach we chose is more robust, suitable for a prototype,
  and sufficient for our evaluation in \Cref{sec:tool:evaluation}.}

\begin{center}
  \begin{minipage}{0.8\linewidth}
    \begin{lstlisting}[language=Ruby,frame=single,rulecolor=\color{gray},numbers=none,basicstyle=\footnotesize\tt]
alias :ddmon, as: GenServer
\end{lstlisting}
  \end{minipage}
\end{center}
\vspace{-1mm}

The module \toolname offers an API equivalent to \texttt{gen\_server}'s
(\texttt{start}, \texttt{start\_link}, etc.). Behind the scenes, \toolname
launches the original \texttt{gen\_server} as a private side process whose
identity (PID) is never revealed outside --- instead, \toolname returns the
monitor's PID. The result is that the monitor ``wraps'' the original \texttt{gen\_server}
and forwards messages similarly to \Cref{fig:example-mon-service}.
Interacting with the monitor is indistinguishable from interacting with the original service ---
although the monitor is a state machine (based on the \texttt{gen\_statem}
behaviour) that implements our algorithm in \Cref{def:deadlock-algo} and the
message forwarding semantics in \Cref{fig:semantics-mon}.
As in our formal model, monitors exchange probes (using distinct \texttt{cast}s).
In addition to a unique identifier (necessary for lock identification in
\Cref{def:deadlock-algo}), each probe \(\Probe\) also stores the list of
services that propagated \(\Probe\) (by \cref{\RuleAlgProbeInPropagate} in
\Cref{fig:algorithm}): when \Toolname reports a deadlock, that list represents a deadlocked set (\Cref{def:deadlock}).
\emph{(For an example of \Toolname deadlock report, see \iftoggle{techreport}{\Cref{lst:envelope-log-deadlock} on page~\pageref{lst:envelope-log-deadlock}}{\cite[Appendix A.1, Listing 2]{TECHREPORT}}.)}
Notably, \Toolname can be configured to emit probes with a specified delay:
we discuss and evaluate the effect of this feature in \Cref{sec:tool:evaluation}.

\subsection{Evaluation}\label{sec:tool:evaluation}

We now evaluate \Toolname in two ways: we estimate its monitoring overhead
(\Cref{sec:tool:evaluation:overhead}) and the trade-off between such overhead
and the deadlock reporting time (\Cref{sec:tool:evaluation:delay}).

\subsubsection{Monitoring Overhead}
\label{sec:tool:evaluation:overhead}

To evaluate the overhead introduced by \Toolname's distributed monitoring, we
measured the number of messages exchanged in randomly-generated systems of
\texttt{gen\_server}-based services, with sizes ranging from 16 to 11\,000
services.\footnote{To generate such random systems we used a dedicated testing DSL described in
  \iftoggle{techreport}{\Cref{app:implementation:testing-dls}}{\cite[Appendix A.1]{TECHREPORT}}.} Such systems non-deterministically reach deadlocks similar to the one
illustrated in \Cref{fig:deadlock-envelope};
in total, approximately 30\% of the runs resulted in deadlocks. \change{change:synth-bench}{Cited other works using synthetic benchmarks}{We adopt synthetic benchmarks as they have been commonly used for evaluating deadlock detection and monitoring performance, e.g. \cite{DBLP:journals/tse/GligorS80,HSNPBM2013,DBLP:journals/sigmod/YeungHL94,DBLP:conf/fase/AcetoAFI21,DBLP:conf/ecoop/AcetoAFI24}.}

\Cref{fig:benchmark-site} shows the relationship between the number of services
and the number of messages exchanged between services; when monitors are
present, the service-to-service messages include probes. \Cref{fig:benchmark-probes} measures the number of probes only. Finally, \Cref{fig:benchmark-sent} measures \emph{all} messages, including the
internal service-to-monitor and monitor-to-service forwarding (depicted in
\Cref{fig:example-mon-service}).
We measure two strategies for emitting probes:

\begin{enumerate}
  \item \emph{Eager}, i.e., monitors emit probes as soon as
    \cref{\RuleAlgQInLocked} in \Cref{fig:algorithm} is applicable.
  \item \emph{Delayed}, i.e., when the conditions of \cref{\RuleAlgQInLocked}
    hold, the monitor does \emph{not} emit a probe \(\Probe\) immediately.
    Instead, the monitor sets a timeout after which emits \(\Probe\) --- but in
    the meanwhile, if the monitor sees a response that matches
    \cref{\RuleAlgRIn} (which deactivates \(\Probe\)), the monitor cancels the
    emission of \(\Probe\).
\end{enumerate}

\changeCameraReady{Revised remark that delayed probe emission has not been proven correct yet}{
Note that in this work we formally prove the correctness of the ``eager'' probe emission strategy above. The ``delayed'' strategy is an optimisation suggested in
previous work on edge-chasing deadlock detection
algorithms~\cite{ChandyMisraHaas83,DBLP:conf/icpads/KimLS97} and adopted e.g.~in~\cite{LingIEEETransComp2006,DBLP:conf/prdc/IzumiDK10}. We believe that the ``delayed'' strategy does not impact monitoring invariants
    (\Cref{def:kic,def:kis}) and preserves the preciseness of deadlock detection
    (\Cref{thm:deadlock-detection-preciseness})
    --- but we have only validated this claim via testing of \Toolname.}

\Cref{fig:benchmark-site,fig:benchmark-probes} show that our ``delayed'' probe
emission strategy tends to reduce the communication overhead by reducing the
number of emitted probes: the reason is that many probes that would be sent by
the ``eager'' strategy are cancelled by a quick response that unlocks
the recipient before the delay expires. \change{change:overhead}{Clarified that the overhead of monitoring is minimal with delayed probes}{
Therefore, if the delay is sufficiently large, the monitoring overhead in service-to-service
communication is barely noticeable: compare the blue and green lines in \Cref{fig:benchmark-site,fig:benchmark-probes}.} \changeCameraReady{Moved and added footnote}{Moreover, each probe is a very small message.}\footnote{In \Toolname, a probe \(\Probe\) is implemented as a unique integer --- but as mentioned above, \(\Probe\) also carries the list of services that propagated \(\Probe\). This list is used to display the deadlocked set, and it could grow up to the length of locked-on dependencies. If necessary, this list can be safely removed to improve performance, as it is irrelevant to the deadlock detection logic.} \Cref{fig:benchmark-sent} shows how the forwarding of every message between
monitor and service (as depicted in \Cref{fig:example-mon-service}) results in
tripling the number of messages sent in total --- although the actual cost of
each forwarding should be small: \change{change:overhead:in-memory}{Clarify cost of message forwarding between monitor and service}{each monitor performs minimal computations, and if a monitor and its overseen service are deployed on the same Erlang VM,
  then their communication is very efficient (usually purely in-RAM without I/O). }

\newcommand{\FIGbenchscale}{0.33}
\begin{figure*}[tbp]
  \centering
  \begin{subfigure}[t]{\FIGbenchscale\textwidth}
    \centering
    \includegraphics[width=\linewidth]{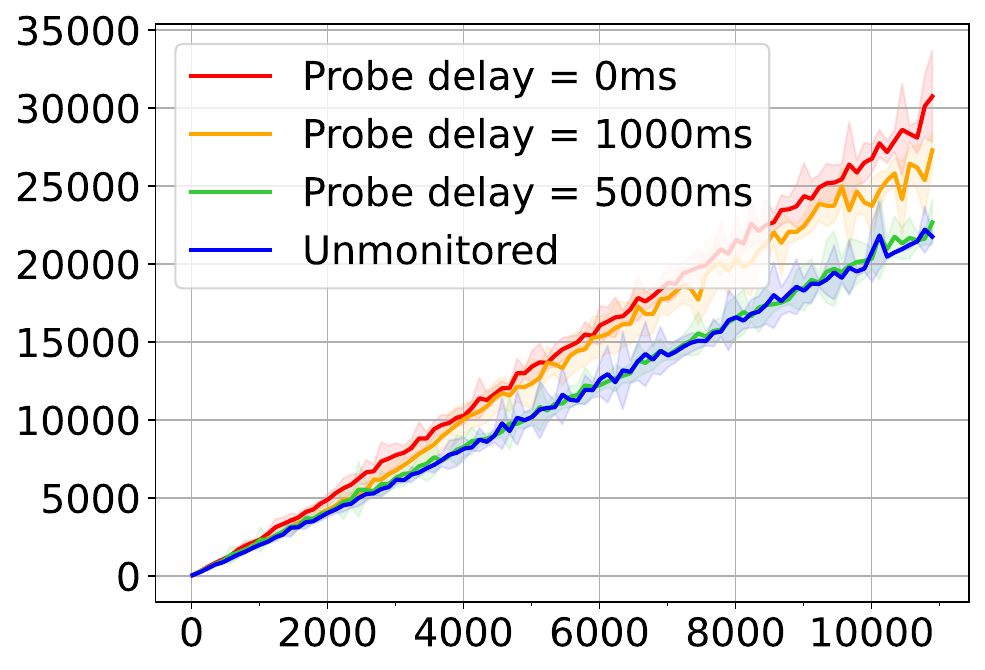}
\caption{Messages between services}\label{fig:benchmark-site}
  \end{subfigure}
  \hfill
  \begin{subfigure}[t]{\FIGbenchscale\textwidth}
    \centering
    \includegraphics[width=\linewidth]{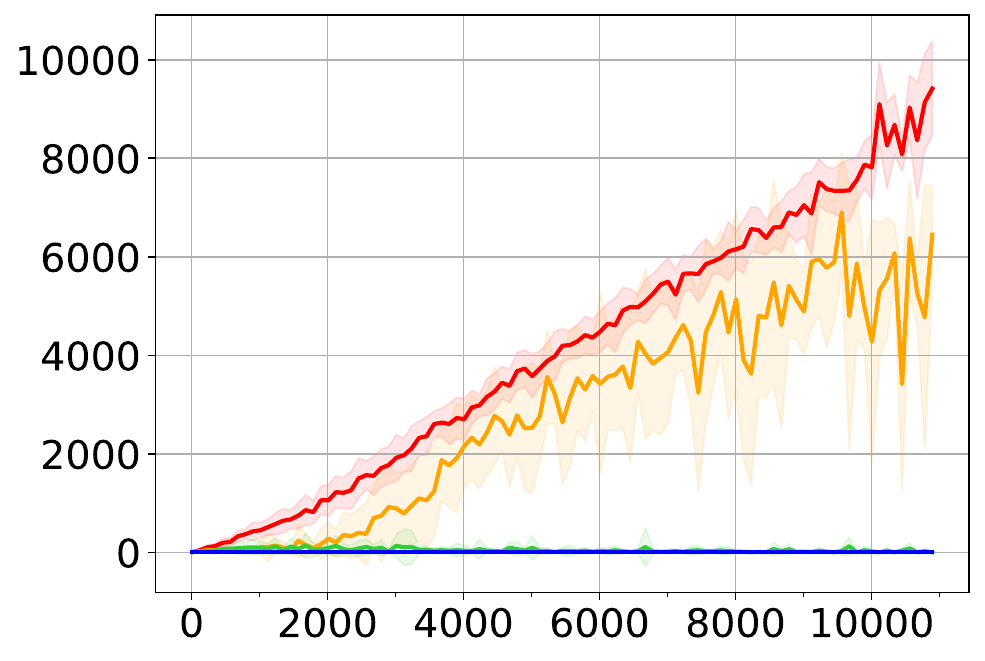}
\caption{Probes only}\label{fig:benchmark-probes}
  \end{subfigure}\hfill
  \begin{subfigure}[t]{\FIGbenchscale\textwidth}
    \centering
    \includegraphics[width=\linewidth]{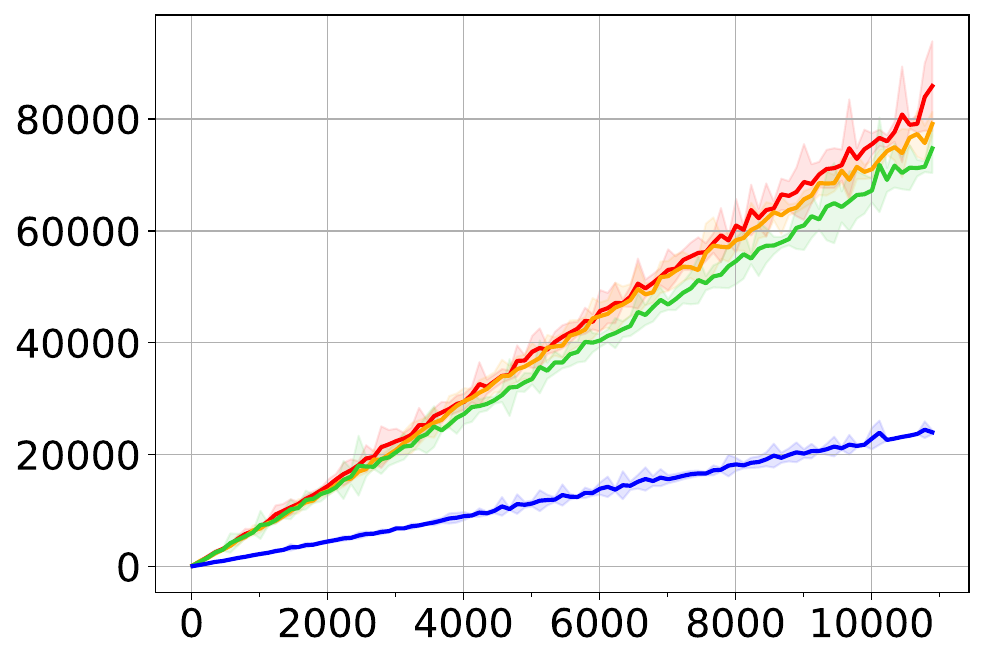}
\caption{All messages}\label{fig:benchmark-sent}
  \end{subfigure}
  \vspace{-2.5mm}\caption{Number of message exchanges (measured on the $y$-axis) vs.~number of
    services in the network ($x$-axis). {\scriptsize(Averages of 20 runs for network sizes < 8500, 10 runs for sizes < 10k, and 5 runs for sizes $\geq$ 10k; standard deviation outlined in the background)}}\label{fig:benchmark}
\end{figure*}

\subsubsection{Monitoring Overhead vs.~Deadlock Reporting Time}
\label{sec:tool:evaluation:delay}

We \changeOursNoMargin{also assess} how delaying the probe emission (via the ``delayed'' strategy
discussed above) impacts the system overhead dynamics, and when deadlocks
are reported.
\Cref{fig:timeseries} compares three executions of monitored systems. In
\Cref{fig:timeseries-delay-0}, the eager emission of probes maintains a
persistent overhead.
In \Cref{fig:timeseries-delay-1000,fig:timeseries-delay-5000}, the probe
emission is delayed by \si{1}{s} and \si{5}{s}: this reduces the probe overhead to around zero
while the system runs without significant locking; when locks last more than the
probe delay, the probe emission increases (this is especially visible in
\Cref{fig:timeseries-delay-1000}). In \Cref{fig:timeseries-delay-5000} a
deadlock occurs after $\sim$\si{3}{s} (noticeable because queries and responses
flatline), but the monitors only raise an alarm at \si{8}{s}.

\todo[nit]{Opportunity to gain space: squeeze plots}

\newcommand{\FIGtimeseriesscale}{0.33}
\begin{figure*}[tbp]
  \centering
  \begin{subfigure}[t]{\FIGtimeseriesscale\textwidth}
    \centering
    \includegraphics[width=\linewidth]{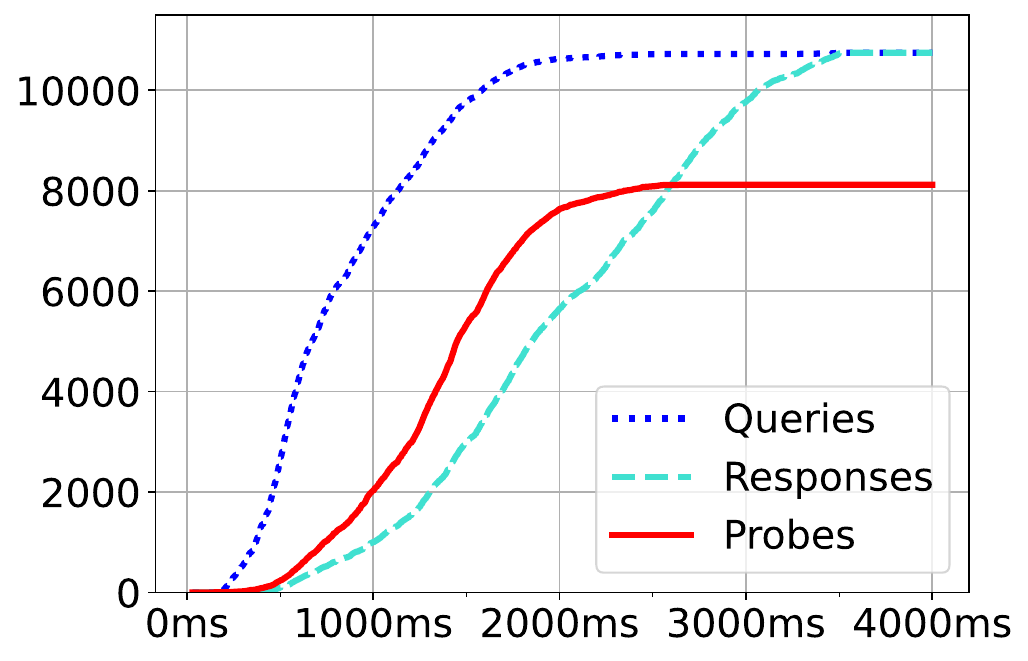}
    \caption{Probes emitted eagerly (red)}\label{fig:timeseries-delay-0}
  \end{subfigure}\hfill
  \begin{subfigure}[t]{\FIGtimeseriesscale\textwidth}
    \centering
    \includegraphics[width=\linewidth]{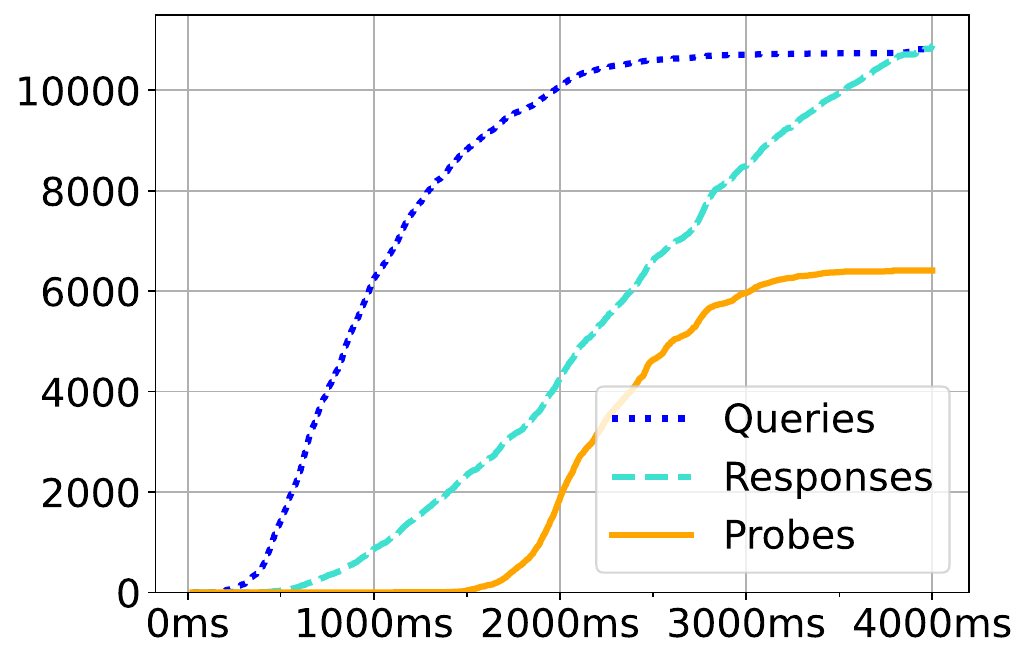}
    \caption{Probe delay = \si{1000}{ms} (orange)}\label{fig:timeseries-delay-1000}
  \end{subfigure}\hfill
  \begin{subfigure}[t]{\FIGtimeseriesscale\textwidth}
    \centering
    \includegraphics[width=\linewidth]{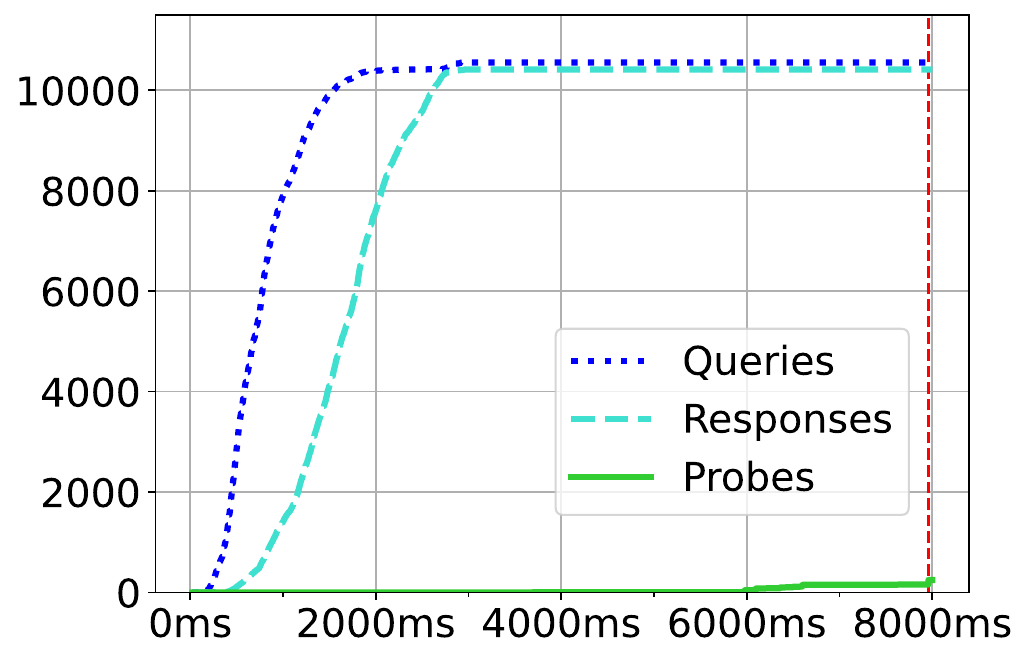}
    \caption{Probe delay = \si{5000}{ms} (green)}\label{fig:timeseries-delay-5000}
  \end{subfigure}\vspace{-2.5mm}\caption{Time series showing executions of monitored systems with 11\,000
    services. The execution in \cref{fig:timeseries-delay-5000}
    ends with a deadlock and its detection is marked with a vertical red line.}
  \label{fig:timeseries}
\end{figure*}
 \changeOurs{change:related-work-overhaul}{Re-organised ``Related work'' section}{\section{Related Work}\label{sec:related}
}

\subsection{Distributed Deadlock Detection Algorithms Based on Edge Chasing}\label{sec:related:edge}

The state of the art in distributed deadlock detection is typically based on the notion of \emph{edge chasing}~\cite{MOSS1985,DBLP:journals/tse/SinhaN85}, used effectively by the seminal work of~\cite{ChandyMisraHaas83,MitchellMerrit84}.
These approaches send probes along the estimated edges of the \emph{wait-for graph}~\cite{DBLP:journals/tse/GligorS80} of the system with the purpose of detecting cycles; different algorithms may send probes in the direction of the wait-for dependency chain or backwards.
The basic ideas and techniques have been improved in terms of correctness~\cite{MitchellMerrit84,DBLP:journals/tse/ChoudharyKST89}, generality~\cite{DBLP:conf/podc/BrachaT84,ChandyMisraHaas83} and performance~\cite{DBLP:journals/tse/SinhaN85,DBLP:journals/sigmod/YeungHL94}.\note{Mention that some work adopts probe delay as optimisation? Which one?}
Compared to this body of work, we show that edge chasing can be leveraged in black-box monitoring, by producing probes and reaching deadlock detections based only on observing incoming and outgoing service messages (see \Cref{fig:example-mon-service}).

Our work distinguishes itself as it contributes:
\begin{enumerate*}
\item\label{it:form}formalised message-passing semantics of
\item\label{it:bbox}black-box monitors for distributed detection, which we
\item\label{it:correct} prove sound and complete,
\item\label{it:mech} via an automated theorem prover.
\end{enumerate*}
None of the distributed deadlock detection algorithms we are aware of address requirement~\ref{it:bbox}.
Moreover, \cite{MitchellMerrit84} provide (non-machine checked) proofs for deadlock detection completeness but not
soundness, potentially allowing for false positives.
Neither of~\cite{DBLP:journals/tse/SinhaN85,DBLP:conf/compsac/IsloorM78,DBLP:journals/tse/ChoudharyKST89,DBLP:journals/tse/GligorS80} proves any correctness results, whereas the proofs in~\cite{MitchellMerrit84,ChandyMisraHaas83,DBLP:journals/tocs/Badal86,DBLP:conf/podc/BrachaT84,DBLP:journals/dpd/SrinivasanR11} are not formal (and even less so mechanisable).
Such proofs are typically based on high-level graph-based arguments without any operational semantics such as ours.

The distributed deadlock detection algorithm closest to ours (\Cref{def:deadlock-algo}) is~\cite{MitchellMerrit84},
which propagates probes backwards similarly to us --- but does not have an equivalent to our \texttt{waiting} set.
As a result, that algorithm requires clients to periodically poll for probe updates from the services they are locked on.
\cite{ChandyMisraHaas83} presents two algorithms: one for \emph{communication (OR) model} and the other for \emph{resource (AND) model}, both of which differ from ours in that they use forward-moving probes.
\change{change:related:and-vs-or}{Better explanation of \emph{AND} and \emph{OR} models}{In the \emph{OR} model, a process (i.e., a service) can send requests to multiple processes and unlocks as soon as it receives a response from any of them; instead, in the \emph{AND} model, a process can be locked on more than one resource and proceeds only when \emph{all} of them are acquired.}
The algorithm for the \emph{OR} model is somewhat closer to ours: when a probe reaches its originator, its propagation is reversed and continues backwards along the wait-for graph (similarly to our~\Cref{def:deadlock-algo}).
Instead, the algorithm for the \emph{AND} model uses forward probes only. \change{change:counter}{Added counterexample to forward-moving probes and decision to change algorithm}{We initially tried to formalise a black-box deadlock monitoring algorithm with forward-moving probes inspired by this approach, but we found a subtle counterexample to its soundness, shown in \Cref{fig:counter}.
Let \(A\), \(B\), \(C\) be monitored services:

\begin{enumerate}
\item Assume that service \(A\) is locked on \(B\); \(B\) is locked on \(C\); \(C\) is working.
\item\label{it:counter:send-probe} \(B\)'s monitor sends a (forward-moving) probe to \(C\).
\item\label{it:counter:send-resp} Meanwhile, service \(C\) sends a response to \(B\).
\item Service \(C\) sends a query to \(A\) (hence, \(C\) becomes locked on \(A\)).
\item \(C\)'s monitor receives the probe from \(B\) (sent at step \ref{it:counter:send-probe}) and propagates it to \(A\).
\item \(A\)'s monitor propagates \(B\)'s probe to \(B\).
\item\label{it:counter:report} \(B\)'s monitor receives its own probe from \(A\); since \(B\) is locked, \(B\)'s monitor reports a deadlock.
\item Service \(B\) receives the response from \(C\) (sent at step \ref{it:counter:send-resp}), so \(B\) is not locked anymore. Therefore, the deadlock reported at step \ref{it:counter:report} was false!
\end{enumerate}
The \emph{AND} model algorithm of \cite{ChandyMisraHaas83} seems to sidestep this issue via some implicit assumption:
e.g., the counterexample above does not arise if probes propagate instantaneously or communications ``freeze'' while probes circulate --- but such assumptions would be unrealistic in a distributed setting. Therefore, that algorithm does not seem obviously adaptable to distributed black-box deadlock monitoring. For this reason, we changed our monitor algorithm design from forward-moving to backward-moving probes (more in the style of~\cite{MitchellMerrit84}, which however was not proven sound, as mentioned above). This
allowed us to formulate~\Cref{def:deadlock-algo} and prove our results.
}

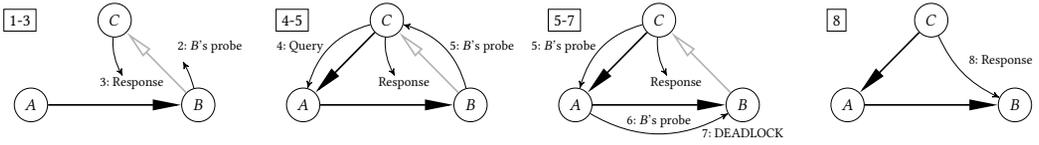
\begin{figure}
  \changeNoMargin{{\scalebox{0.63}{\hspace{-6mm}\definecolor{lgray}{RGB}{180,180,180}
\newcommand{\countersetdim}{\useasboundingbox (-2.7,-2.5) rectangle (2.7,0.5);}
\newcommand{\countercap}[1]{\node[draw] at (-2,0) {#1};}
\begin{tabular}{cccc}
\begin{tikzpicture}[->,>=stealth',auto,node distance=2.5cm, semithick,ser/.style={circle, draw, minimum size=20pt},every label/.style={align=left},every node/.style={align=center}]
  \countersetdim
  \countercap{1-3}

  \node[ser] (C) [] {\(C\)};
  \node[ser] (A) [below left of=C] {\(A\)};
  \node[ser] (B) [below right of=C] {\(B\)};

\path[
  draw=black, fill=black, line width=1pt, minimum width=10pt, minimum height=10pt,
  -{Triangle[length=6mm, bend]}
  ]
  (A) edge node {} (B);

  \path[
  draw=lgray, fill=white, line width=1pt, minimum width=10pt, minimum height=10pt,
  -{Triangle[length=6mm, bend]}
  ]
  (B) edge node {} (C);

  \path
  (B) edge[bend right, shorten >=15mm] node[right] {\footnotesize 2: \(B\)'s probe} (C)
  (C) edge[bend right=50, shorten >=15mm] node[below=-10pt] {\footnotesize 3: Response} (B);
\end{tikzpicture}
&
\begin{tikzpicture}[->,>=stealth',auto,node distance=2.5cm, semithick,ser/.style={circle, draw, minimum size=20pt},every label/.style={align=left},every node/.style={align=center}]
  \countersetdim
  \countercap{4-5}

  \node[ser] (C) [] {\(C\)};
  \node[ser] (A) [below left of=C] {\(A\)};
  \node[ser] (B) [below right of=C] {\(B\)};

\path[
  draw=black, fill=black, line width=1pt, minimum width=10pt, minimum height=10pt,
  -{Triangle[length=6mm, bend]}
  ]
  (A) edge node {} (B)
  (C) edge node {} (A);

  \path[
  draw=lgray, fill=white, line width=1pt, minimum width=10pt, minimum height=10pt,
  -{Triangle[length=6mm, bend]}
  ]
  (B) edge node {} (C);

  \path
  (C) edge[bend right] node[left] {\footnotesize 4: Query} (A)
  (B) edge[bend right] node[right] {\footnotesize 5: \(B\)'s probe} (C)
  (C) edge[bend right=50, shorten >=15mm] node[below=-10pt] {\footnotesize Response} (B)
  ;
\end{tikzpicture}
&
\begin{tikzpicture}[->,>=stealth',auto,node distance=2.5cm, semithick,ser/.style={circle, draw, minimum size=20pt},every label/.style={align=left},every node/.style={align=center}]
  \countersetdim
  \countercap{5-7}

  \node[ser] (C) [] {\(C\)};
  \node[ser] (A) [below left of=C] {\(A\)};
  \node[ser,label=below:{\footnotesize 7: DEADLOCK}] (B) [below right of=C] {\(B\)};

\path[
  draw=black, fill=black, line width=1pt, minimum width=10pt, minimum height=10pt,
  -{Triangle[length=6mm, bend]}
  ]
  (A) edge node {} (B)
  (C) edge node {} (A);

  \path[
  draw=lgray, fill=white, line width=1pt, minimum width=10pt, minimum height=10pt,
  -{Triangle[length=6mm, bend]}
  ]
  (B) edge node {} (C);

  \path
  (C) edge[bend right] node[left] {\footnotesize 5: \(B\)'s probe} (A)
  (A) edge[bend right] node[above] {\footnotesize 6: \(B\)'s probe} (B)
  (C) edge[bend right=50, shorten >=15mm] node[below=-10pt] {\footnotesize Response} (B)
  ;
\end{tikzpicture}
&
\begin{tikzpicture}[->,>=stealth',auto,node distance=2.5cm, semithick,ser/.style={circle, draw, minimum size=20pt},every label/.style={align=left},every node/.style={align=center}]
  \countersetdim
  \countercap{8}

  \node[ser] (C) [] {\(C\)};
  \node[ser] (A) [below left of=C] {\(A\)};
  \node[ser] (B) [below right of=C] {\(B\)};

\path[
  draw=black, fill=black, line width=1pt, minimum width=10pt, minimum height=10pt,
  -{Triangle[length=6mm, bend]}
  ]
  (A) edge node {} (B)
  (C) edge node {} (A);

  \path
  (C) edge[bend right=20] node {\footnotesize 8: Response} (B)
  ;
\end{tikzpicture}
\end{tabular}
\hspace{-4mm}
 }}
\vspace{-3mm}\caption{\changeNoMargin{Counterexample for a distributed deadlock monitoring algorithm with forward-moving probes, with a deadlock reported
      for a non-deadlocked service. The numbers indicate the order of events as listed in \Cref{sec:related:edge}. The straight black arrows depict the wait-for graph. The straight gray arrow
      shows that $B$ is locked on $C$, but an incoming response from $C$ is about to unlock $B$. Bent arrows show the message
      flow.}}\label{fig:counter}
  }\end{figure}

\subsection{Monitoring and Other Approaches to Deadlock Detection}

\change{change:cite-hyperprop}{Cite recent work on monitoring hyper-properties}{\noindent There exists recent research on monitoring hyper-properties, although none of it addresses deadlock detection specifically.
For instance,~\cite{Finkbeiner2019,DBLP:conf/forte/AcetoAAF22} provide frameworks for monitoring formulae in different variants of HML logic, but do not discuss correctness of particular applications.
In~\cite{DBLP:conf/concur/AcetoAAFGW24}, the authors synthesise decentralised monitors from hyper-property specifications and guarantee that these monitors are deadlock-free by construction.
Our work achieves this result through transparency (\Cref{lem:net-transp-compl}), but the main target is the detection of deadlocks in the \emph{system being monitored}.
}

The GoodLock algorithm~\cite{DBLP:conf/spin/Havelund00} is one of the first works to propose the detection of potential deadlocks of Java threads \emph{at runtime};
this influential work is extended in \cite{DBLP:conf/hvc/BensalemH05,DBLP:conf/issta/BensalemFHM06} and combined with static analysis techniques in \cite{DBLP:conf/hvc/AgarwalWS05}.
More recently, \cite{DBLP:conf/ifm/AlbertGI16} combined static analysis and testing (often considered as a runtime analysis technique) for deadlock detection.
Neither of the above approaches are fully black-box.  Moreover, no formal soundness or completeness guarantees are provided.
The authors of~\cite{HSNPBM2013} develop an algorithm for the runtime deadlocks detection of MPI programs.
The approach is mostly black box and its scalability is demonstrated through an extensive empirical evaluation.
The work \cite{ISPDC2018} develops a runtime analysis tool to detect concurrency bugs in embedded software, with deadlocks being one of the classes of bugs detected.
The tool mainly relies on trace log analysis and, in that sense, the approach may be considered as black-box.
Both~\cite{ISPDC2018,HSNPBM2013} are however centralised and are \emph{not} backed up by any soundness or completeness guarantees.

\subsection{Proof Mechanisation for Deadlock Detection and Prevention}

There are efforts that mechanise correctness proofs for deadlocks in distributed systems, but these are considerably different from our work.
E.g., \cite{DBLP:conf/pldi/WilcoxWPTWEA15,DBLP:conf/nsdi/HanceHMP21} consider proofs for static deadlock prevention as opposed to deadlock detection at runtime.
More closely to our work, \cite{DBLP:conf/vortex/AudritoH23} present a Coq formalisations for monitors that do not communicate with each other, thereby limiting their analysis to protocol violations, not deadlock detection (which necessarily relies on probes).

The work of~\cite{DBLP:journals/toplas/CogumbreiroHMY19} presents a method and tool for runtime deadlock detection for systems with distributed barrier synchronisations, such as X10~\cite{CharlesX10OOPSLA2005}.
Their approach employs decentralised runtime entities called distributed barrier programs that can be seen as monitors, but it is unclear to what extent the implementation is (or can be made to be) black-box since this was not one of the requirements of the work.
Their deadlock analysis relies on a different abstraction to wait-for graphs called \emph{task-event graphs}.
This is accompanied by mechanised proofs in Coq showing that detections based on task-event graphs are sound and complete with respect to actual deadlocks.
Despite the similarities, these proofs present graph-based arguments whereas our proofs are closer to the actual implementation, verifying the operational behaviour and verdicts of our decentralised monitors with respect to the deadlocks occurring in the network.
For instance, proofs of monitor transparency do not feature in the mechanised proofs of~\cite{DBLP:journals/toplas/CogumbreiroHMY19}.

\note{There are some notes about related work commented out here}%
 \section{Conclusion and Future Work}
\label{sec:conclusion}

We presented the first formal theory of distributed black-box monitors for
deadlock detection in RPC-based systems. We introduced the general correctness
\Cref{crit:transparency,crit:detection-preciseness} (transparency and detection
preciseness) and we proved that our monitors and algorithm
(\Cref{def:deadlock-algo}) satisfy them
(\Cref{lem:net-transp-compl,lem:net-transp-sound,thm:deadlock-detection-preciseness}).
Our results are mechanised in Coq. We also presented and evaluated \Toolname, a
practical \changeCameraReady{Added for consistency}{prototype} tool based on our theoretical results.

As future work, we plan to instantiate our generic monitoring framework
(\Cref{sec:monitoring}) with other deadlock detection algorithms (e.g.~having
probe propagation strategies closer
to~\cite{ChandyMisraHaas83,MitchellMerrit84}) investigating whether they are sound and
complete. If so, we can easily add such algorithms to \Toolname, and evaluate
and compare their performance.
\change{change:future:delay}{Mention formalisation of delayed probes}{We plan to formalise and prove the correctness of the delayed probe emission strategy used in \Cref{sec:tool:evaluation:overhead}:
  we do not expect conceptual difficulties, although the mechanisation may require a large number of small changes.}
\change{change:future:fairness}{Mention mechanisation of fairness}{We also plan to mechanise \emph{fairness of components}~\cite{Glabbeek_2019} and thus the final part of~\Cref{thm:detect-completeness-upgrade}
  --- although mechanising fairness appears challenging and we are not aware of other attempts in the literature.}

\change{change:and-or-preliminary}{Mention preliminary results of experiments with AND and OR models}{
We have preliminary results indicating that our algorithm can support the \emph{AND}
resource model of \cite{ChandyMisraHaas83} (i.e., sending multiple queries and then awaiting all responses)
with minimal adjustments. The \emph{OR} model is more challenging: it is
not obvious how to correctly identify ``\emph{OR} queries'' with a black-box monitoring
approach; after this design question is solved, its formalisation will likely
lead to different fundamental definitions and invariants w.r.t.~the algorithm in
this paper. We have already tested some ideas and a preliminary implementation
of a deadlock monitoring algorithm (inspired by~\cite{ChandyMisraHaas83}), but the formalisation is
future work.}

We also plan to build upon our results to develop a theory for distributed
deadlock detection via \emph{observing monitors}, which only receive
\emph{notifications} about messages sent/received by a service (akin to the
\texttt{trace} mechanism in Erlang/OTP~\cite{erlang-doc-trace}).
The advantages of this approach are \changeCameraReady{Toning down a bit}{potentially} lower overhead w.r.t.~our proxy monitors, as well as transparency by construction.
The disadvantage is that the delays introduced by notifications can
cause monitor states and probes to become subtly outdated, leading to false or
missed deadlock reports. We have already found counterexamples showing that our
algorithm (\Cref{def:deadlock-algo}) and~\cite{MitchellMerrit84} become
unsound if na\"ively used in this setting. Therefore, we expect that
designing a distributed deadlock detection algorithm for observing monitors that
can be proven sound and complete will be a non-trivial effort.

\change{change:timeouts}{Mention handling timeouts as future work}{Another topic worth investigating is how to handle RPC calls with timeouts. For that, we need to
  address a key design question: how can a black-box monitor become aware of a timeout occurring
  within the overseen service? If the monitor is somehow notified about the timeout, then we believe
  our algorithm can be adapted and remain sound and complete (albeit its ``black-box'' nature might be
  questioned); otherwise, fundamental definitions of ``lock'' and ``deadlock'' may need to be
  changed, and sound and complete black-box monitoring may turn out to be impossible.}

\note{There are some notes about the conclusion commented out here}%
 
\section*{Acknowledgements}
We thank Denis Drobny and Motorola Solutions Danmark A/S for providing examples and for the insightful discussion that started this research.
This work was partially supported by the Independent Research Fund Denmark project \emph{``Hyben''} and by the University of Malta Research Seed Fund grant no.~CPSRP01-25 \emph{``Runtime Monitoring for IoT.''}

\pagebreak
\section*{Data Availability Statement}

The companion artifact of this work is \Toolname~\cite{rowicki_2025_16909304}, a distributed deadlock
detection monitoring tool for Erlang and Elixir programs based on the
\texttt{gen\_server} behaviour. \Toolname is described in
\Cref{sec:tool:implementation} and \iftoggle{techreport}{\Cref{app:implementation}}{\cite[Appendix A]{TECHREPORT}}. The artifact
replicates the experiments and measurements described in
\Cref{sec:tool:evaluation}. 

As supplementary material, the Coq mechanisation of our main results is
available in~\cite{rowicki_2025_16909482}. The mechanised results are marked
throughout the paper with the symbol \coqed.

\bibliography{main}

\iftoggle{techreport}{\newpage
    {\huge{Appendices}}

    \appendix

    \section{Additional Information about \toolname}
\label{app:implementation}

In Elixir applications, a programmer can alias \texttt{gen\_server} with
\toolname as shown in \Cref{fig:instr-code} (top lines; the rest of the figure
shows unchanged application code).

\noindent
\begin{minipage}{\textwidth}
  \centering
  \begin{minipage}{0.485\linewidth}
\begin{lstlisting}[language=Ruby,xrightmargin=0pt,framexrightmargin=0pt]
# Adding this line instruments
# the service with our monitors.
alias :ddmon, as: GenServer

# No further changes are needed.
# (Standard gen_server-based code follows)

def start(args) do
  GenServer.start(__MODULE__, args, [])
end


def init(_args) do
  x = GenServer.start(:some_service, 42)
  {:ok, x}
end
\end{lstlisting}
  \end{minipage}
  \captionof{figure}{Example usage of \toolname in Elixir.}\label{fig:instr-code}\end{minipage}

The instrumentation of \toolname is a bit more verbose in Erlang w.r.t.~Elixir,
because Erlang lacks Elixir's \texttt{alias} feature: this is shown in
\Cref{fig:instr-code-erlang} (compare with \Cref{fig:instr-code}). To better support Erlang, \Toolname could adopt the instrumentation strategies
discussed in \cref{footnote:alternative-instrumentation}, or an automatic AST
rewriting to replace every use of \texttt{gen\_server} APIs with the
corresponding \Toolname APIs. (At the time of writing, we have not yet
implemented these facilities.)

\noindent
\begin{minipage}{\textwidth}
  \centering
  \begin{minipage}{0.5\linewidth}
\begin{lstlisting}[language=Erlang,escapechar=\&,xrightmargin=0pt]
start(Args) ->
    gen_server:start(?MODULE, Args, []).

init(_Args) ->
    X = gen_server:call(some_service, 42),
    {ok, X}.

%============%  ...becomes  %============%

start(Args) ->
    ?ddmon:start(?MODULE, Args, []).

init(_Args) ->
    X = ?ddmon:call(some_service, 42),
    {ok, X}.
\end{lstlisting}
\end{minipage}
  \captionof{figure}{Example of usage of \toolname in Erlang.}
  \label{fig:instr-code-erlang}
\end{minipage}

\subsection{Testing Framework and DSL}
\label{app:implementation:testing-dls}

\Toolname comes with a generic server implementation which interprets queries as expressions of a
simple DSL\@. The DSL allows programming entire behaviours of the system by specifying topologies of
nested \emph{query sessions}. Session is a collection of all queries that were sent as a consequence
of a single initial call from outside the network; for example, in \Cref{fig:deadlock-envelope}
three sessions can be distinguished: (1,4,7), (2,5,8) and (3,6,9). Each query contains information
on how its session should be executed further, most importantly which other services to call and
with what instructions. This design allows defining various scenarios very flexibly, without the
need of implementing each service individually.

\Cref{lst:envelope-scenario} shows how the example from \Cref{fig:deadlock-envelope} is implemented
as a test case for \toolname. Services are identified by integers, of which 0, 2 and 4 correspond to
endpoints (\(E\)), while 1, 3 and 5 to proxies (\(P\)). The program starts three similarly
programmed sessions: for example, \texttt{s1} begins with a query to the service number \texttt{0}
instructing it to first wait for some unspecified time and then call the service number \texttt{3},
ordering it to send its query to \texttt{5}. If \texttt{5} receives it while not busy with a
different query, it replies immediately to \texttt{3} as there are no more instructions to follow.

To assist in diagnosing and fixing deadlocks, \toolname allows to optionally collect logs using the Erlang/OTP
\texttt{trace} mechanism, and visualise them as shown in \Cref{lst:envelope-log-deadlock}, where
services fell into the deadlock as pictured in \Cref{fig:deadlock-envelope}. Each column lists
events happening in each service chronologically from top to bottom: services exchange queries (e.g.
\texttt{M5~!~Q(s3)}) and monitors conclude that they are locked and activate fresh probes
(\texttt{=>~LOCK(P0)}). As monitors exchange probes (\texttt{M1~!~P3}), one of them eventually
receives its own active probe (\texttt{?~P3} in \texttt{M3}) and reports a deadlock.

\Cref{lst:more-tool-logs} shows two logs from the same application, with a
non-deterministic deadlock that occurs in one of the two runs.

\begin{lstlisting}[language=Erlang, escapechar=\&,caption={Scenario from \Cref{fig:deadlock-envelope} implemented as a test case for \toolname.},label={lst:envelope-scenario}]
{sessions,
 [ {s1, { 0               % Start the session with a call to service 0.
        , { {wait, 0, 10} % Order the service to wait randomly between 0 and 10ms.
          , [1, 3]        % Make it call service 1, so that it calls service 3.
          }}}             % When there is nothing to do, a reply is sent.
 , {s2, { 2, { {wait, 0, 10}, [3, 5] }}} % Ditto with different services.
 , {s3, { 4, { {wait, 0, 10}, [5, 0] }}} ]}.
\end{lstlisting}

\noindent
\begin{minipage}{\textwidth}
  \lstlog{
\begin{lstlisting}[language=Log,xrightmargin=0pt,xleftmargin=2.5em,framexleftmargin=0pt,framexleftmargin=0pt,caption={Log produced by \toolname on a \texttt{gen\_server}-based application based on \Cref{fig:deadlock-envelope-code}. Column headers are service names,
      and each column shows events in chronological order. Queries (\qtag) are marked with tags (\texttt{s1, s2, s3}) to visualise their causality. This run ends with a deadlock, which is reported. Notice that the last line reports the exact set of deadlocked services,
      which in this example does \emph{not} involve services \texttt{M2} nor \texttt{M4}:
      those services are locked on the deadlocked services, but are not deadlocked themselves.},label={lst:envelope-log-deadlock} ]
[ M0 ]:     | [ M1 ]:     | [ M2 ]:     | [ M3 ]:     | [ M4 ]:     | [ M5 ]:
------------|-------------|-------------|-------------|-------------|-------------
            |             |             |             | M5 ! Q(s3)  |
            |             |             |             | => LOCK(P0) |
            |             |             |             |             | M4 ? Q(s3)
            |             |             |             |             | M0 ! Q(s3)
            |             |             |             |             | => LOCK(P1)
M5 ? Q(s3)  |             |             |             |             |
            |             | M3 ! Q(s2)  |             |             |
            |             | => LOCK(P2) |             |             |
            |             |             | M2 ? Q(s2)  |             |
            |             |             | M5 ! Q(s2)  |             |
            |             |             | => LOCK(P3) |             |
            |             |             |             |             | M3 ? Q(s2)
M1 ! Q(s1)  |             |             |             |             |
=> LOCK(P4) |             |             |             |             |
            | M0 ? Q(s1)  |             |             |             |
            | M3 ! Q(s1)  |             |             |             |
            | => LOCK(P5) |             |             |             |
            |             |             | M1 ? Q(s1)  |             |
            |             |             | M1 ! P3     |             |
            |  ? P3       |             |             |             |
            | M0 ! P3     |             |             |             |
 ? P3       |             |             |             |             |
M5 ! P3     |             |             |             |             |
            |             |             |             |             | ? P3
            |             |             |             |             | M4 ! P3
            |             |             |             |             | M3 ! P3
            |             |             |             | ? P3        |
            |             |             | ? P3
### DEADLOCK ### (M3 -> M5 -> M0 -> M1 -> M3)
\end{lstlisting}
  }\end{minipage}

\noindent
\begin{minipage}{\textwidth}
  \begin{minipage}[t]{0.49\linewidth}
    \lstlog{
\begin{lstlisting}[language=Log,xrightmargin=0pt,xleftmargin=2.5em,framexleftmargin=0pt,framexleftmargin=0pt,caption={Successful execution.}]
| [ M0 ]:     | [ M1 ]:     | [ M2 ]:
|-------------|-------------|-------------
| I0 ? Q(i0)  |             |
|             | I1 ? Q(i1)  |
| M1 ! Q(i0)  |             |
| => LOCK(P0) |             |
|             | M0 ? Q(i0)  |
|             | M2 ! Q(i1)  |
|             | => LOCK(P1) |
|             | M0 ! P1     |
|             |             | M1 ? Q(i1)
|  ? P1       |             |
| I0 ! P1     |             |
|             |             | M1 ! R(i1)
|             | => UNLOCK   |
|             | I1 ! R(i1)  |
|             | M0 ! R(i0)  |
| =>  UNLOCK  |             |
| I0 ! R(i0)  |             |
|             |             | I2 ? Q(i2)
|             |             | M0 ! Q(i2)
|             |             | => LOCK(P2)
| M2 ? Q(i2)  |             |
| M2 ! R(i2)  |             |
|             |             | => UNLOCK
|             |             | I2 ! R(i2)
### TERMINATED ###
\end{lstlisting}}\end{minipage}\hfill
  \begin{minipage}[t]{0.49\linewidth}
\lstlog{
\begin{lstlisting}[language=Log,xrightmargin=0pt,xleftmargin=1em,framexleftmargin=0pt,framexleftmargin=0pt,numbers=none,caption={Deadlocking execution.}]
| [ M0 ]:     | [ M1 ]:     | [ M2 ]:
|-------------|-------------|-------------
| I0 ? Q(i0)  |             |
|             | I1 ? Q(i1)  |
|             |             | I2 ? Q(i2)
| M1 ! Q(i0)  |             |
| => LOCK(P0) |             |
|             | M0 ? Q(i0)  |
|             | M2 ! Q(i1)  |
|             | => LOCK(P1) |
|             | M0 ! P1     |
|             |             | M1 ? Q(i1)
|             |             | M0 ! Q(i2)
|             |             | => LOCK(P2)
|             |             | M1 ! P2
|  ? P1       |             |
| M2 ? Q(i2)  |             |
| I0 ! P1     |             |
| M2 ! P0     |             |
|             |             |  ? P0
|             |             | I2 ! P0
|             |             | M1 ! P0
|             |  ? P0       |
|             | I1 ! P0     |
|             | M0 ! P0     |
|  ? P0       |             |
| => ### DEADLOCK ### (M0 -> M1 -> M2 -> M0)
\end{lstlisting}}\end{minipage}
  \captionof{figure}{Two \Toolname logs from two executions of the same \texttt{gen\_server}-based application, with two different results.}
  \label{lst:more-tool-logs}\end{minipage}

 \section{Implementing Replicated SRPC Services}\label{app:replicated}

\definecolor{workbusycolor}{rgb}{1,0.8,0.8}
\definecolor{workfreecolor}{rgb}{0.9,1.0,0.9}
\definecolor{workcolor}{rgb}{0.9,0.9,0.9}
\definecolor{replcolor}{rgb}{0.95,0.95,0.95}
\definecolor{clientcolor}{rgb}{1,1,1}

\newcommand{\repltikz}[1]{
  \begin{tikzpicture}[->,>=stealth', shorten >=2pt, auto, node distance=2cm, semithick,work/.style={circle, fill=workcolor, draw, minimum size=10pt, node distance=1cm},repl/.style={circle, fill=replcolor, draw, node distance=2cm, minimum size=10pt},client/.style={circle, fill=clientcolor, draw, node distance=2cm, minimum size=10pt},initedge/.style={gray},every label/.style={align=left},every node/.style={align=center}]

    \node[repl] (r) {\(R\)};

    \node[client] (i0) [above left of=r] {\(C_0\)};
    \node[client] (i1) [below left of=r] {\(C_1\)};

    \ifthenelse{\equal{#1}{\string 0}}
    {
      \node[work, fill=workfreecolor, node distance=1.7cm] (w0) [above right of=r] {\(W_0\)};
      \node[work, fill=workbusycolor, node distance=1.7cm] (w1) [right of=r] {\(W_1\)};
      \node[work, fill=workfreecolor, node distance=1.7cm] (w2) [below right of=r] {\(W_2\)};

      \path
      (i0) edge[bend left] node[near start] {1} (r) (r) edge[bend left] node[near start, below left] {2 {\scriptsize(\(W_0\))}} (i0) (i0) edge[bend left] node[near start] {3} (w0) ;
    }{}
    \ifthenelse{\equal{#1}{\string 1}}
    {
      \node[work, fill=workbusycolor, node distance=1.7cm] (w0) [above right of=r] {\(W_0\)};
      \node[work, fill=workbusycolor, node distance=1.7cm] (w1) [right of=r] {\(W_1\)};
      \node[work, fill=workfreecolor, node distance=1.7cm] (w2) [below right of=r] {\(W_2\)};

      \path
      (i1) edge[bend left] node[near start] {4} (r) (r) edge[bend left] node[near start, right] {5 {\scriptsize(\(W_2\))}} (i1) (i1) edge[bend right] node[near start] {6} (w2) ;
    }{}
    \ifthenelse{\equal{#1}{\string 2}}
    {
      \node[work, fill=workfreecolor, node distance=1.7cm] (w0) [above right of=r] {\(W_0\)};
      \node[work, fill=workfreecolor, node distance=1.7cm] (w1) [right of=r] {\(W_1\)};
      \node[work, fill=workbusycolor, node distance=1.7cm] (w2) [below right of=r] {\(W_2\)};

      \path
      (w0) edge[bend left] node[near start] {6} (i0) (i0) edge[out=240, in=170, dashed] node[left, near start] {7} (r) ;
    }{}
  \end{tikzpicture}
}

Despite its streamlined design, our SRPC service model is sufficiently
expressive to capture \emph{service replication}, where one logical
\emph{replicating server} distributes queries among a number of \emph{workers}.
A client of such a server needs to first obtain the name of an available worker
by querying the server. The server replies with the first free worker in its
pool, and if all are busy, picks one of them. The client then queries the
received worker directly; after receiving a response, the client releases the
worker by sending a cast to the replicating server.

\medskip \scalebox{0.8}{\begin{minipage}{1.2\linewidth}
\repltikz{0}\hfill
\repltikz{1}\hfill
\repltikz{2}\end{minipage}
}\medskip 

The diagrams above illustrate this on an example involving a replicating server
\(R\) with three workers \(W_{0,1,2}\) and two clients \(C_{0,1}\). Initially,
the worker \(W_1\) is busy, while \(W_0\) and \(W_2\) are available. The client
\(C_0\) calls \(R\) and receives a reference to \(W_0\), to which it sends its
query. \(R\) considers \(W_0\) busy. In the meantime, \(C_1\) also requests a
worker, and gets assigned \(W_2\). Finally, \(C_0\) receives a response from
\(W_0\) and sends a cast to \(R\), so that it considers \(W_0\) available again.

Notice that replicated services can be used to attenuate deadlocks, but they
cannot always prevent them.
A deadlock caused by circular queries involving services $S_1,\ldots,S_n$ may be
avoided if the deadlocking RPC query from \(S_n\) to \(S_1\) is assigned to
service \(S'_1\) that is a replica of \(S_1\).
However, even if \(S_1\) has \(m\) replicas, there is a chance that the request
from \(S_n\) is assigned to \(S_1\) anyway --- e.g., if all replicas are
busy.

\section{Proof of \Cref{thm:detect-completeness-upgrade}}
\label{app:proof-completeness-upgraded}

\newcommand{\mpathLength}[1]{\operatorname{length}({#1})}\change{change:proof:eventual-deadlock}{Proof updated with a simpler distance measure, and clarified.}{From the statement and proof of \Cref{thm:deadlock-detection-preciseness} we
  know that there is a monitor-flushing path \(\mpath'\) (constructed by
  \Cref{thm:alarm-condition}) that begins with \(\MNet'\) and reaches a network
  state where a monitor reports a deadlock; also, in \(\MNet'\) there is an
  active probe \(\Probe\) that causes a deadlock alarm as \(\Probe\) reaches
  that monitor through \(\mpath'\). The monitor-flushing path \(\mpath'\) built
  by \Cref{thm:alarm-condition} has the shape:

  \smallskip \centerline{\(\mpath' \;=\; \MAct_1 \cdot \ldots \cdot \MAct_m \cdot \mpath''_1 \cdot \ldots \cdot \mpath''_k\qquad \text{where \(\forall i \in 1..k\):}\quad \mpath''_i \;=\; \mnetcomm{\Name_i}{\Name'_i}{\Probe} \cdot \MAct'_{1,i} \cdot \ldots \MAct'_{j_i,i}
\)}\smallskip 

  \noindent where \(\MAct_i\) are monitor-flushing actions of the monitor that has the probe \(\Probe\) in its
  monitor queue (those actions do \emph{not} forward \(\Probe\)). In other words, the
  path \(\mpath'\) says that the active probe \(\Probe\) will be forwarded \(k\) times
  (once per sub-path \(\mpath''_i\), for \(i \in 1..k\)) after \(m\)
  monitor-flushing actions occur (which empty the monitor queue containing \(\Probe\)); each
  forwarding may be followed by more actions (denoted as \(\MAct'_{j,i}\) above).
We can then define \(\mpathLength{\mpath'}\) as the pair \((k, m)\): intuitively, this represents the distance between the active probe \(\Probe\) and the monitor that owns (i.e., created) \(\Probe\).

  We then prove that such a distance never increases as the network \(\MNet'\) runs,
  and can be decreased toward \((0, 0)\). More precisely, for any \(\MNet''\) reachable from \(\MNet'\), if we take the
  monitor-flushing path \(\mpath''\) constructed by \Cref{thm:alarm-condition}
  in \(\MNet''\), we have \(\mpathLength{\mpath''} \leq \mpathLength{\mpath'}\)
  by standard lexicographic ordering. Moreover, in any such \(\MNet''\) where \(\mpathLength{\mpath''} > (0, 0)\),
  there is always a monitor with an
  enabled monitor-flushing action \(\MAct\) such that \(\MNet'' \trans{\MAct}
  \MNet'''\) and, if we take the  monitor-flushing path \(\mpath'''\) constructed
  by \Cref{thm:alarm-condition} in \(\MNet'''\), we have \(\mpathLength{\mpath'''} < \mpathLength{\mpath''}\).

  Hence, assuming fairness of components~\cite{Glabbeek_2019}, every
  monitor that reduces the current monitor-flushing path constructed by
  \Cref{thm:alarm-condition} will be eventually scheduled to execute an action.
  Therefore, the length of the path constructed by
  \Cref{thm:alarm-condition} will eventually reach \((0, 0)\), and the
  monitor that owns \(\Probe\) will raise an alarm.
  \emph{(Note: this last paragraph from ``Hence\ldots'' is not part of our Coq mechanisation.)}
} }{}

\end{document}